\documentclass[a4paper, USenglish, thm-restate, numberwithinsect]{lipics-v2021}

\pdfoutput=1 

\hideLIPIcs  

\usepackage{cite}
\usepackage{hyperref}
\usepackage{tikz}
\usetikzlibrary{arrows,calc,fit,shapes, automata,backgrounds,positioning}

\usepackage{pgfplots}

\pgfplotsset{width=5cm,compat=1.10}
\usepgfplotslibrary{fillbetween}

\usepackage{xcolor}
\usepackage{stmaryrd}
\usepackage{tikz}
\usetikzlibrary{automata, positioning, arrows, backgrounds}

\newcommand{\N}{\mathbb{N}}

\newcommand{\Z}{\mathbb{Z}}

\newcommand{\Q}{\mathbb{Q}}

\newcommand{\vect}[1]{\mathbf{#1}}

\definecolor{niceredbright}{HTML}{bd0310}
\definecolor{nicebluebright}{HTML}{197b9b}
\definecolor{nicered}{HTML}{7f0a13}
\definecolor{niceblue}{HTML}{104354}
\definecolor{nicegreen}{HTML}{217516}
\definecolor{nicepurple}{HTML}{884bab}
\definecolor{nicebg}{HTML}{f6f0e4}
\definecolor{niceredlight}{HTML}{c9888d}
\definecolor{nicebluelight}{HTML}{78a4b8}
\definecolor{nicegreenlight}{HTML}{76de68}
\definecolor{nicepurplelight}{HTML}{bc87db}

\DeclareMathOperator{\interior}{int}

\DeclareMathOperator{\FO}{FO}
\DeclareMathOperator{\Fill}{Fill}

\DeclareMathOperator{\Dir}{dir}
\DeclareMathOperator{\dir}{\Dir}
\DeclareMathOperator{\lin}{lin}
\DeclareMathOperator{\invariant}{distance}
\DeclareMathOperator{\Reach}{Reach}
\DeclareMathOperator{\VectorSpace}{VectSp}
\DeclareMathOperator{\Vectorspace}{\VectorSpace}

\title{Geometry of Reachability Sets of Vector Addition Systems} 

\author{Roland Guttenberg}{Technical University of Munich, Germany}{guttenbe@in.tum.de}{0000-0001-6140-6707}{}

\author{Mikhail Raskin}{LaBRI, University of Bordeaux, France}{mikhail.raskin@u-bordeaux.fr}{0000-0002-6660-5673}{}

\author{Javier Esparza}{Technical University of Munich, Germany}{esparza@in.tum.de}{0000-0001-9862-4919}{}

\authorrunning{Roland Guttenberg, Mikhail Raskin, Javier Esparza} 

\Copyright{Roland Guttenberg, Mikhail Raskin, Javier Esparza} 



\category{} 

\funding{The project has received funding from the European Research Council under the European Union’s Horizon 2020 research and innovation programme under grant agreement No 787367.}

\nolinenumbers 

\EventEditors{Guillermo A. P\'{e}rez and Jean-Fran\c{c}ois Raskin}
\EventNoEds{2}
\EventLongTitle{34th International Conference on Concurrency Theory (CONCUR 2023)}
\EventShortTitle{CONCUR 2023}
\EventAcronym{CONCUR}
\EventYear{2023}
\EventDate{September 18--23, 2023}
\EventLocation{Antwerp, Belgium}
\EventLogo{}
\SeriesVolume{279}
\ArticleNo{6}

\begin{document}

\maketitle

\begin{abstract}
Vector Addition Systems (VAS), aka Petri nets, are a popular model of concurrency. The reachability set of a VAS is the set of configurations reachable from the initial configuration. Leroux has studied the geometric properties of VAS reachability sets, and used them to derive decision procedures for important analysis problems. In this paper we continue the geometric study of reachability sets. We show that every reachability set admits a finite decomposition into disjoint almost hybridlinear sets enjoying nice geometric properties. Further, we prove that the decomposition of the reachability set of a given VAS is effectively computable. As a corollary, we derive a new proof of Hauschildt's 1990 result showing the decidability of the question whether the reachability set of a given VAS is semilinear. As a second corollary, we prove that the complement of a reachability set, if it is infinite, always contains an infinite linear set.
\end{abstract}

\section{Introduction}


Vector Addition Systems (VAS), also known as Petri nets, are a popular model of concurrent systems.
The VAS reachability problem consists of deciding if a target configuration of a VAS is reachable from some initial configuration. It was proved decidable in the 1980s \cite{Mayr81, Kosaraju82}, but its complexity (Ackermann-complete) could only be determined recently \cite{CzerwinskiLLLM19, CzerwinskiO21, Leroux21}. 

The \emph{reachability set} of a VAS is the set of all configurations reachable from the initial configuration. Configurations are tuples of natural numbers, and so the reachability set of a VAS is a subset of $\N^n$ for some $n$ called the \emph{dimension} of the VAS. Results on the geometric properties of reachability sets have led to new algorithms in the past. For example, in \cite{Leroux12} it was shown that every configuration outside the reachability set $\vect{R}$ of a VAS is separated from $\vect{R}$ by a semilinear inductive invariant. This immediately leads to an algorithm for the reachability problem consisting of two semi-algorithms, one enumerating all possible paths to certify reachability, and one enumerating all semilinear sets and checking if they are separating inductive invariants. Another example is \cite{Leroux13}, where it was shown that semilinear reachability sets are flatable. The result led to an algorithm for deciding whether a semilinear set is included in or equal to the reachability set of a given VAS. 

The separability and flatability results of \cite{Leroux12,Leroux13} are proven not only for VAS reachability sets, but for arbitrary semilinear \emph{Petri sets}, a larger class with a geometric definition introduced in \cite{Leroux12}. So, in particular, \cite{Leroux13} is an investigation into the geometric structure of semilinear Petri sets. In this paper we study the structure of the \emph{non-semilinear} Petri sets. We introduce hybridization, or, equivalently, the class of \emph{almost hybridlinear} sets, a generalization of the hybridlinear sets introduced by Ginsburg and Spanier \cite{ginsburg1963bounded} and further studied by Chistikov and Haase \cite{ChistikovH16}. We prove the following decomposition:

\begin{restatable}{theorem}{TheoremFinalPartition}
Let \(\vect{X}\) be a Petri set. For every semilinear set \(\vect{S}\) there exists a partition \(\vect{S}=\vect{S}_1 \cup \dots \cup \vect{S}_k\) into pairwise disjoint full linear sets such that for all \(i \in \{1, \ldots, k \}\) either \(\vect{X} \cap \vect{S}_i=\emptyset\), \(\vect{S}_i \subseteq \vect{X}\) or \(\vect{X} \cap \vect{S}_i\) is irreducible with hybridization $\vect{S}_i$. Further, if \(\vect{X}\) is the reachability set of a VAS, then the partition is computable. \label{TheoremFinalPartition}
\end{restatable}

Defining strong hybridization and irreducibility is beyond the scope of this introduction; in fact, they will be introduced in Section \ref{Sec:NewPeriodicityProperty} and \ref{Sec:PetriSetsAndDimension} of this paper. However, we can already explain two properties of the irreducible sets with a strong hybridization which, combined with Theorem \ref{TheoremFinalPartition}, have important consequences. 

Firstly, irreducible sets with hybridization are always non-semilinear. This leads to a simple algorithm for deciding whether the reachability set $\vect{X} \subseteq \N^d$ of a given VAS of dimension $d$ is semilinear. 
Let $\vect{S}:=\N^d$ and compute the partition $\vect{S}_1 \cup \dots \cup \vect{S}_k$ of Theorem \ref{TheoremFinalPartition}. For every $1 \leq i \leq k$, check whether $\vect{X} \cap \vect{S}_i=\emptyset$ or $\vect{S}_i \subseteq \vect{X}$ hold\footnote{It is well known that the first question can be reduced to the VAS reachability problem, and the second is decidable by the flatability results mentioned before.}. If this is the case for all $i$, then let \(J\) be the set of indices \(i\), where \(\vect{S}_i \subseteq \vect{X}\) holds. We have $\bigcup_{i \in J} \vect{S}_i = \vect{X}  \cap \vect{S} = \vect{X}$, and so, since $\vect{S}_1, \ldots, \vect{S}_k$ are linear,  $\vect{X}$ is semilinear. Otherwise, by Theorem \ref{TheoremFinalPartition} there exists an $i$ such that $\vect{X} \cap \vect{S}_i$ is irreducible with hybridization $\vect{S}_i$, and hence non-semilinear. Since semilinear sets are closed under intersection, $\vect{X}$ is not semilinear. 

The decidability of the semilinearity of VAS reachability sets was first proved by Hauschildt \cite{Hauschildt90}, and in fact we arrive at essentially the same algorithm. However, we provide a simpler correctness proof and a clear geometric intuition. Further, our theorem holds for arbitrary Petri sets, a larger class than VAS reachability sets. 


Secondly, if a set \(\vect{X}\) is irreducible with hybridization \(\vect{S}\), then there are infinitely many points in the boundary \(\partial \vect{S}\) of \(\vect{S}\) that do not belong to $\vect{S}$, i.e., \ \(|\partial \vect{S} \setminus \vect{X}|=\infty\). This allows to prove that if \(\vect{S} \setminus \vect{X}\) is infinite, then \(\vect{S} \setminus \vect{X}\) contains an infinite linear set, which was left as a conjecture in \cite{JancarLS19}. Namely the proof is now a simple induction on the dimension of the semilinear set \(\vect{S}\): If \(\vect{S} \setminus \vect{X}\) is infinite, then some \(\vect{S}_i \setminus \vect{X}\) is infinite. If for this \(i\), we have \(\vect{X} \cap \vect{S}_i=\emptyset\) or \(\vect{S}_i \subseteq \vect{X}\), then \(\vect{S}_i \setminus \vect{X}\) is semilinear and hence contains an infinite line. Otherwise we have that \(|\partial \vect{S}_i \setminus \vect{X}|=\infty\), and hence by induction \(\partial \vect{S}_i \setminus \vect{X}\) contains an infinite line. This corollary is a first step towards understanding the complements of VAS reachability sets, for which little is known. 

The sections of the paper follow the structure of the main theorem. Section \ref{Sec:Preliminaries} contains preliminaries. Section \ref{Sec:DirectionsOfAPeriodicSet} introduces smooth sets, preparing for the introduction of strong hybridization and Petri sets in Section \ref{Sec:NewPeriodicityProperty}. Section \ref{Sec:PetriSetsAndDimension} introduces irreducibility and proves Theorem \ref{TheoremFinalPartition}. Section \ref{SectionFinalPartitionAndCorollaries} proves the corollaries of Theorem \ref{TheoremFinalPartition}.

\section{Preliminaries} \label{Sec:Preliminaries}


We let $\N, \mathbb{Z}, \mathbb{Q}, \mathbb{Q}_{\geq 0}$ denote the natural, integer, and (non-negative) rational numbers.

Furthermore, we use uppercase letters except \(A\) for sets, with \(A\) being used for matrices. We use boldface for vectors and sets of vectors. We denote the cardinality of a set \(\vect{X}\) as \(|\vect{X}|\).

Given sets \(\vect{X},\vect{Y} \subseteq \mathbb{Q}^n, Z \subseteq \mathbb{Q}\), we write \(\vect{X}+\vect{Y}:=\{\vect{x}+\vect{y} \mid \vect{x} \in \vect{X}, \vect{y} \in \vect{Y}\}\) and \(Z \cdot \vect{X}:=\{\lambda \cdot \vect{x} \mid \lambda \in Z, \vect{x} \in \vect{X}\}\). By identifying elements \(\vect{x}\in \mathbb{Q}^n\) with \(\{\vect{x}\}\), we define \(\vect{x}+\vect{X}:=\{\vect{x}\}+\vect{X}\), and similarly \(\lambda \cdot \vect{X}:=\{\lambda\} \cdot \vect{X}\) for \(\lambda \in \mathbb{Q}\). We denote by \(\vect{X}^C\) the complement of \(\vect{X}\). On \(\mathbb{Q}^n\), we consider the usual Euclidean norm and its generated topology. We denote the closure of a set \(\vect{X}\) in this topology by \(\overline{\vect{X}}\).

A \emph{vector space} \(\vect{V} \subseteq \mathbb{Q}^n\) is a set such that \(\vect{0} \in \vect{V}\), \(\vect{V}+\vect{V} \subseteq \vect{V}\) and \(\Q \cdot \vect{V} \subseteq \vect{V}\). Given a set \(\vect{F} \subseteq \mathbb{Q}^n\), the vector space generated by \(\vect{F}\) is the smallest vector space containing \(\vect{F}\). Every vector space \(\vect{V}\) is \emph{finitely generated} (f.g.), i.e.\ there exists a finite set \(\vect{F} \subseteq \Q^n\) generating \(\vect{V}\). Furthermore, it can also be expressed as \(\{\vect{x}\in \Q^n \mid A \vect{x}=0\}\) for some integer matrix \(A\). 

\subsection{Cones, lattices, and periodic sets}
A set \(\vect{C} \subseteq \mathbb{Q}^n\) is a \emph{cone} if \(\vect{0}\in \vect{C}\), \(\vect{C}+\vect{C} \subseteq \vect{C}\) and \(\Q_{>0}\vect{C} \subseteq \vect{C}\). Given a set \(\vect{F} \subseteq \mathbb{Q}^n\), the cone generated by \(\vect{F}\) is the smallest cone containing \(\vect{F}\). If \(\vect{C}\) is a cone, then \(\vect{C}-\vect{C}\) is the vector space generated by \(\vect{C}\). Not every cone is finitely generated (f.g.). Instead, we have:

\begin{lemma}{\cite[Corollary 7.1a]{LinearProgramming}}
Let \(\vect{C} \subseteq \Q^n\) be a cone. Then \(\vect{C}\) is finitely generated if and only if \(\vect{C}=\{\vect{x} \in \vect{C}-\vect{C} \mid A \vect{x} \geq \vect{0}\}\) for some integer matrix \(A\). \label{LemmaFinitelyGeneratedCones}
\end{lemma}

In particular, finitely generated cones are closed. The \emph{interior} of a finitely generated cone \(\vect{C}\) is the set \(\interior(\vect{C})=\{\vect{x} \in \vect{C}-\vect{C} \mid A \vect{x} > \vect{0}\}\), where \(A\) is a matrix as above. The boundary of the cone is \(\partial(\vect{C}):=\overline{\vect{C}} \setminus \interior(\vect{C})\).  It is well known that the boundary of a cone is a a finite union of lower dimensional cones, called facets\cite{LinearProgramming}. In fact, there is a defining matrix \(A\) such that the facets are exactly the sets of solutions obtained by changing one of the inequalities of $A \vect{x} \geq 0$ into an equality. For example, the left part of Figure \ref{Fig:ExampleFullLinearSet} shows the cone \(\{(x,y) \mid x-y \geq 0, y \geq 0\}\). Its facets are the sets \(\{(x,y)\mid x-y=0, y \geq 0\}\) and \(\{(x,y)\mid x \geq y, y =0\}\) (shown as black lines in the picture), and their union is the boundary of the cone. 


A cone \(\vect{C}\) is \emph{definable} if it is definable in \(\FO(\mathbb{Q}, +, \geq)\).  A cone \(\vect{C}\) is definable if{}f  \(\vect{C} \setminus \{\vect{0}\}=\{\vect{x} \in \vect{C}-\vect{C} \mid A_1 \vect{x} >\vect{0}, A_2 \vect{x}\geq \vect{0}\}\) for some integer matrices \(A_1, A_2\). In this case the closure \(\overline{\vect{C}}\) is finitely generated. Intuitively, changing an equation from from \(\geq 0\) to \(>0\) removes a facet. Removing all facets yields \(\interior(\vect{C})\).

\begin{figure}[h!]
\noindent%
\begin{minipage}{4.5cm}
\begin{tikzpicture}
\begin{axis}[
    axis lines = left,
    xlabel = { },
    ylabel = { },
    xmin=0, xmax=10,
    ymin=0, ymax=12,
    xtick={0,2,4,6,8,10},
    ytick={0,2,4,6,8,10,12},
    ymajorgrids=true,
    xmajorgrids=true,
    thick,
    smooth,
    no markers,
]

\addplot+[
    name path=A,
    color=red,
    ]
    coordinates {
    (0,0)(10,10)
    };

\addplot+[
    color=red,
    name path=B,
    ]
    coordinates {
    (0,0)(10,0)
    };
    
\addplot[
    color=black,
    very thick,
]
coordinates {
(0,0)(10,10)
};

\addplot[
    color=black,
    very thick,
]
coordinates {
(0,0.1)(10,0.1)
};
    
\addplot[
    color=red!40,
]
fill between[of=A and B];

\addplot[
    fill=blue,
    fill opacity=0.5,
    only marks,
    ]
    coordinates {
    (0,0)(0,2)(0,4)(0,6)(0,8)(0,10)(0,12)(2,0)(2,2)(2,4)(2,6)(2,8)(2,10)(2,12)(4,0)(4,2)(4,4)(4,6)(4,8)(4,10)(4,12)(6,0)(6,2)(6,4)(6,6)(6,8)(6,10)(6,12)(8,0)(8,2)(8,4)(8,6)(8,8)(8,10)(8,12)(10,0)(10,2)(10,4)(10,6)(10,8)(10,10)(10,12)(12,0)(12,2)(12,4)(12,6)(12,8)(12,10)(12,12)
    };
    
\end{axis}
\end{tikzpicture}
\end{minipage}%
\begin{minipage}{4.5cm}
\begin{tikzpicture}
\begin{axis}[
    axis lines = left,
    xlabel = { },
    ylabel = { },
    xmin=0, xmax=5,
    ymin=0, ymax=15,
    xtick={0,1,2,3,4,5},
    ytick={0,3,6,9,12,15},
    ymajorgrids=true,
    xmajorgrids=true,
    thick,
    smooth,
    no markers,
]

\addplot[
    color=blue,
]
{3 * x};


\addplot+[
    color=red,
    name path=A,
    domain=2:5, 
]
{3};

\addplot+[
    color=red,
    name path=B,
    domain=2:5, 
]
{3*x-3};

\addplot[
    color=red!40,
]
fill between[of=A and B];

\addplot[
    fill=blue,
    fill opacity=0.5,
    only marks,
    ]
    coordinates {
    (0,0)(1,0)(1,2)(1,3)(2,0)(2,2)(2,3)(2,4)(2,5)(2,6)(3,0)(3,2)(3,3)(3,4)(3,5)(3,6)(3,7)(3,8)(3,9)(4,0)(4,2)(4,3)(4,4)(4,5)(4,6)(4,7)(4,8)(4,9)(4,10)(4,11)(4,12)(5,0)(5,2)(5,3)(5,4)(5,5)(5,6)(5,7)(5,8)(5,9)(5,10)(5,11)(5,12)(5,13)(5,14)(5,15)
    };

\end{axis}
\end{tikzpicture}
\end{minipage}
\quad
\begin{minipage}{6cm}
\begin{tikzpicture}[-triangle 60,auto, thick]
\tikzstyle{every state}=[rectangle,thick,draw=blue!75,fill=blue!20,minimum 
     size=6mm,text=black,minimum width=6mm]
\newcommand*{\distancesubx}{1.5cm}
\newcommand*{\distancesuby}{1.5cm}
\newcommand*{\distancelabel}{0.07cm}

\node[state] (A0) at (0,\distancesuby) {Per. \(\vect{P}\)};
\node[state] (B0) at (\distancesubx,0) {Cone \(\vect{C}\)};
\node[state] (B1) at (\distancesubx,\distancesuby) {\(\cap\)};
\node[state] (B2) at (\distancesubx,2*\distancesuby) {Lat. \(\vect{L}\)};
\node[state] (C0) at (2*\distancesubx,\distancesuby) {\(\VectorSpace\)};
\path (A0) edge[] node[below left=\distancelabel] {\(\Q_{\geq 0}\vect{P}\)} (B0);
\path (A0) edge[] node[above left=\distancelabel] {\(\vect{P}-\vect{P}\)} (B2);
\path (B0) edge[] node[below right=\distancelabel] {\(\vect{C}-\vect{C}\)} (C0);
\path (B2) edge[] node[above right=\distancelabel] {\(\Q_{\geq 0}\vect{L}\)} (C0);
\path (B0) edge[] (B1);
\path (B2) edge[] (B1);
\path (B1) edge[] (A0);
\end{tikzpicture}
\end{minipage}
\caption{\textit{Left}: The cone generated by \(\{(1,1), (1,0)\}\) is shown in red, with its boundary in black. The lattice \((2,0)\mathbb{Z}+(0,2)\mathbb{Z}\) is the set of of blue dots. Their intersection is the periodic set \(\{(2,0), (2,2)\}^{\ast}\). \newline
\textit{Middle}: The periodic set \(\vect{P}=\{(1,0),(1,2),(1,3)\}^{\ast}\) is shown in blue. Intuitively, the set $\{(1,1), (2,1), (3, 1), \ldots \}$ is a ``hole'' of $\vect{P}$. Inside \(\vect{P}\) we find the red area $(2, 3) + \vect{P}$, whose blue points do not intersect the hole, i.e., \((2,3)+\Fill(\vect{P})\subseteq \vect{P}\). \newline
\textit{Right}: Graph comparing the classes of sets defined in Section \ref{Sec:Preliminaries}.}
\label{Fig:ExampleFullLinearSet}
\end{figure}
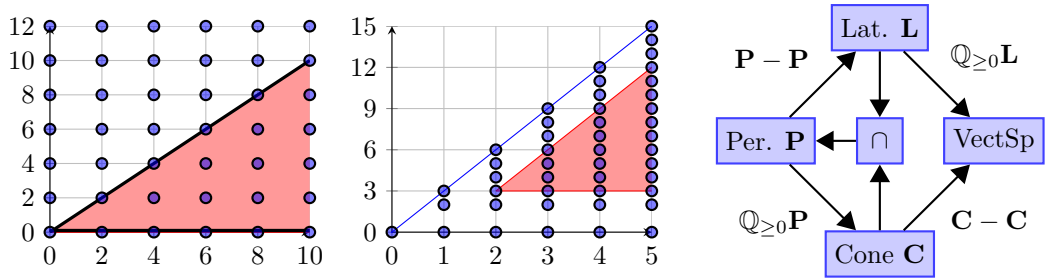

A set \(\vect{L} \subseteq \mathbb{Z}^n\) is a \emph{lattice} if \(\vect{L}+\vect{L} \subseteq \vect{L}, -\vect{L} \subseteq \vect{L}\) and \(\vect{0}\in \vect{L}\). For any finite set \(\vect{F}=\{\vect{x}_1, \dots, \vect{x}_s\} \subseteq \N^n\), the lattice generated by \(\vect{F}\) is \(\Z \vect{x}_1 + \dots + \Z \vect{x}_s\). Every lattice is finitely generated, and even has a generating set linearly independent over \(\Q\).

A set \(\vect{P} \subseteq \mathbb{N}^n\) is a \emph{periodic set} if \(\vect{P} + \vect{P} \subseteq \vect{P}\) and \(\vect{0} \in \vect{P}\). For any set \(\vect{F} \subseteq \N^n\), the periodic set \(\vect{F}^{\ast}\) generated by \(\vect{F}\) is the smallest periodic set containing \(\vect{F}\). We have \(\vect{F}^{\ast}=\{\vect{p}_1+\dots+\vect{p}_r \mid r\in \N, \vect{p}_i \in \vect{F} \text{ for all } i\}\). A periodic set \(\vect{P}\) is \emph{finitely generated} if \(\vect{P}=\vect{F}^{\ast}\) for some finite set \(\vect{F}\). Finitely generated periodic sets are characterized as follows:

\begin{lemma}{\cite[Lemma V.5]{Leroux13}}
Let \(\vect{P} \subseteq \N^n\) be a periodic set. Then \(\vect{P}\) is finitely generated as a periodic set if and only if \(\mathbb{Q}_{\geq 0}\vect{P}\) is finitely generated as a cone. \label{LemmaCharactizeFinitelyGeneratedPeriodicSets}
\end{lemma}

Any set generates a lattice, a cone and a vector space. In the case of periodic sets these have simple formulas; namely \(\vect{P}-\vect{P}\), as well as \(\mathbb{Q}_{\geq 0}\vect{P}\) and \(\VectorSpace(\vect{P}):=\mathbb{Q}_{\geq 0}(\vect{P}-\vect{P})=\mathbb{Q}_{\geq 0}\vect{P} - \mathbb{Q}_{\geq 0}\vect{P}\) respectively. These are also depicted in the right of Figure \ref{Fig:ExampleFullLinearSet}. On the other hand, if \(\vect{C}\) is a cone and \(\vect{L}\) is a lattice, then \(\vect{C} \cap \vect{L}\) is a periodic set. We will consider periodic sets of this form in more depth in Section \ref{SectionDefinitionFull}.

\subsection{Dimension}

The \emph{dimension} of a vector space defined as its number of generators is a well-known concept. It can be extended to arbitrary subsets of \(\mathbb{Q}^n\) as follows.

\begin{definition}{\cite{Leroux11, Leroux12}}
Let \(\vect{X} \subseteq \mathbb{Q}^n\). The \emph{dimension} of \(\vect{X}\), denoted \(\dim(\vect{X})\), is the smallest natural number \(k\) such that there exist finitely many vector spaces \(\vect{V}_i \subseteq \mathbb{Q}^n\) with \(\dim(\vect{V}_i)\leq k\) and vectors \(\vect{b}_i \in \mathbb{Q}^n\) such that \(\vect{X} \subseteq \bigcup_{i=1}^r \vect{b}_i + \vect{V}_i\). 
\end{definition}

This dimension function has the following properties.

\begin{restatable}{lemma}{BasicDimensionProperties}
Let \(\vect{X}, \vect{X}' \subseteq \mathbb{Q}^n, \vect{b}\in \mathbb{Q}^n\). Then \(\dim(\vect{X})=\dim(\vect{b}+\vect{X})\) and 
\(\dim(\vect{X} \cup \vect{X}')=\max \{\dim(\vect{X}), \dim(\vect{X}')\}\). Further, if \(\vect{X} \subseteq \vect{X}'\), then \(\dim(\vect{X}) \leq \dim(\vect{X}')\). \label{BasicDimensionProperties}
\end{restatable}

\begin{lemma}{\cite[Lemma 5.3]{Leroux11}}
Let \(\vect{P}\) be periodic. Then \(\dim(\vect{P})=\dim(\Vectorspace(\vect{P}))\). \label{LemmaFromJerome}
\end{lemma}

Lemma \ref{LemmaFromJerome} for example shows that the lattice and the cone depicted in the left of Figure \ref{Fig:ExampleFullLinearSet}, as well as the periodic set obtained as intersection have dimension \(2\), because all of them generate the vector space \(\mathbb{Q}^2\).

\subsection{Finitely generated vs. full periodic sets} \label{SectionDefinitionFull}

A set \(\vect{L}\) is \emph{linear} if \(\vect{L}=\vect{b}+\vect{P}\) with \(\vect{b}\in \N^n\) and \(\vect{P} \subseteq \N^n\) a finitely generated periodic set. A set \(\vect{S}\) is \emph{semilinear} if it is a finite union of linear sets. The semilinear sets coincide with the sets definable via formulas \(\varphi \in \FO(\mathbb{N}, +, \geq)\), also called Presburger Arithmetic. This is the usual definition of a linear set in theoretical computer science, however, we will work with a slightly smaller class of linear sets, which we call full linear sets. As shown for example in \cite{Woods15}, working with this smaller class does not change the class of semilinear sets: A set \(\vect{S}\) is semilinear if and only if it is a finite union of full linear sets,\ i.e. linear sets \(\vect{b}+\vect{P}\) where \(\vect{P}\) is not only finitely generated, but even full, as in the following definition.

\begin{definition}
A periodic \(\vect{P}\) is \emph{full} if \(\vect{P}=\vect{C} \cap \vect{L}\), where \(\vect{C}\) is a f.g. cone and \(\vect{L}\) a lattice.%
\end{definition}

Full linear sets have even been used as the main definition of linear set in the literature before, for example in \cite{NguyenP18}. Furthermore, while not directly defined, this class was also utilized in \cite{Leroux12, Leroux13} as well. For an example of a finitely generated periodic set which is not full, consider the middle of Figure \ref{Fig:ExampleFullLinearSet}.

There is another equivalent definition of full periodic sets, which uses an overapproximation of a periodic set we call \(\Fill(\vect{P})\). This overapproximation was first introduced in \cite{Leroux11} with the terminology \(\lin(\vect{P})\). However, we avoid this terminology because in \cite{Leroux12, Leroux13}, the same author used the same notation with a slightly different meaning.

\begin{definition}
\label{def:fill}
Let $\vect{P}$ be a periodic set. The \emph{fill} of $\vect{P}$ is the set \(\Fill(\vect{P}):=(\vect{P}-\vect{P}) \cap \overline{\mathbb{Q}_{\geq 0}\vect{P}}\).%
\end{definition}

Intuitively, we overapproximate \(\vect{P}\) via the intersection of the obvious lattice and cone. The reason for using the closure of $\mathbb{Q}_{\geq 0}\vect{P}$ instead of the cone $\mathbb{Q}_{\geq 0}\vect{P}$ itself is Lemma \ref{LemmaCharactizeFinitelyGeneratedPeriodicSets}: If the cone is not closed, then the periodic set, in our case \(\Fill(\vect{P})\), is not finitely generated. If \(\vect{P}\) was already finitely generated, the definitions coincide.

\begin{restatable}{lemma}{LemmaEquivalentDefinitionOfFull}
A periodic set \(\vect{P}\) is full if and only if \(\overline{\Q_{\geq 0}\vect{P}}\) is a f.g. cone and \(\vect{P}=\Fill(\vect{P})\). \label{LemmaEquivalentDefinitionOfFull}
\end{restatable}

By Lemma \ref{LemmaCharactizeFinitelyGeneratedPeriodicSets}, full periodic sets are finitely generated: Namely, their cone \(\Q_{\geq 0}\vect{P}\) equals \(\vect{C} \cap \Q_{\geq 0}\vect{L}\), which as intersection of f.g. cones is finitely generated by Lemma \ref{LemmaFinitelyGeneratedCones}.

Let us conclude this subsection with the main advantage of full linear over linear sets.

\begin{restatable}{lemma}{LemmaFullThenLessPeriods}
Let \(\vect{P},\vect{Q}\) periodic, \(\vect{P}\) full, \(\vect{b},\vect{c} \in \Q^n\) such that \(\vect{c}+\vect{Q} \subseteq \vect{b}+\vect{P}\). Then \(\vect{Q} \subseteq \vect{P}\).\label{LemmaFullThenLessPeriods}%
\end{restatable}

\begin{proof}
Since \(\vect{P}\) is full, by Lemma \ref{LemmaEquivalentDefinitionOfFull} it is sufficient to prove \(\vect{Q} \subseteq \vect{P}-\vect{P}\) and \(\vect{Q} \subseteq \overline{\mathbb{Q}_{\geq 0}\vect{P}}\). 

To prove \(\vect{Q} \subseteq \vect{P}-\vect{P}\), observe that \(\vect{Q}=(\vect{c}+\vect{Q}) - \vect{c}  \subseteq (\vect{b}+\vect{P})-(\vect{b}+\vect{P})=\vect{P}-\vect{P}\). 

To prove \(\vect{Q} \subseteq \overline{\mathbb{Q}_{\geq 0}\vect{P}}\), write \(\overline{\Q_{\geq 0}\vect{P}}=\{\vect{x} \in \VectorSpace(\vect{P}) \mid A \vect{x} \geq 0\}\) for a matrix \(A\), as in Lemma \ref{LemmaFinitelyGeneratedCones}. Let \(A_k\) be the \(k\)-th row of \(A\). It suffices to show \(A_k \vect{x} \geq 0\) for all \(\vect{x} \in \vect{Q}\). If we had \(A_k \vect{x} <0\), then \(A_k (\vect{c}+\lambda \vect{x}) < A_k \vect{b}\) for large enough \(\lambda\), contradicting \(\vect{c}+\vect{Q} \subseteq \vect{b}+\vect{P}\).
\end{proof}

Observe that if we replace full by finitely generated, then the lemma does not hold: Choose \(\vect{P}\) as the periodic set in the middle of Figure \ref{Fig:ExampleFullLinearSet}, then \((2,3)+\{(1,1)\}^{\ast} \subseteq \vect{P}\), and the property is violated, since \((1,1) \not \in \vect{P}\). 

Another advantage is that many proofs simplify in the full case. The following such case will be a cornerstone of our main algorithm:

\begin{restatable}{lemma}{LemmaRemovingInnerConeReducesDimension}
\cite[Corollary D.3]{Leroux13} Let \(\vect{P}\) be a finitely generated periodic set. For every \(\vect{x} \in \vect{P}\) the set \(\vect{S}:=\vect{P} \setminus (\vect{x} + \vect{P})\) is semilinear and satisfies \(\dim(\vect{S}) < \dim(\vect{P})\). \label{LemmaRemovingInnerConeReducesDimension}
\end{restatable}

To prove this, first show that \(\vect{P}\) contains \(\vect{v}+\Fill(\vect{P})\), as in the middle of Figure \ref{Fig:ExampleFullLinearSet}, and reduce to the case of full periodic \(\vect{P}\). For full \(\vect{P}\) it is geometrically clear; for example removing the red cone in the middle of Figure \ref{Fig:ExampleFullLinearSet} from the set, we are left with a finite union of lines.

\section{Smooth Periodic Sets} \label{Sec:DirectionsOfAPeriodicSet}


Not all periodic sets we need in the paper are finitely generated, but they are smooth, a class 
introduced by Leroux in \cite{Leroux13}.  Intuitively, a smooth set $\vect{P}$ is ``close'' to being finitely generated, in the sense that $\Fill(\vect{P})$ is finitely generated. This result (very similar to a result of \cite{Leroux11}) is proven in  
Section \ref{subsec:SmoothImpliesFinitelyGenerated}. In the rest of the section we show that smooth sets satisfying a novel condition are closed under intersection and enjoy good properties (Proposition \ref{PropositionIntersectionPeriodicSets}).

We first reintroduce the set of directions of a periodic set. 

\begin{definition}{\cite{Leroux13}}
\label{def:directions}
Let \( \vect{P} \) be a periodic set.  A vector \(\vect{d} \in \Q^n\) is a \emph{direction} of  \(\vect{P}\) if there exists \(m \in \N_{>0}\) and a point \(\vect{x}\) such that \(\vect{x}+\N \cdot m \vect{d} \subseteq \vect{P}\), i.e. some line in direction \(\vect{d}\) is fully contained in \(\vect{P}\). The set of directions of \(\vect{P}\) is denoted \(\dir(\vect{P})\).
\end{definition}

We can now define smooth periodic sets.

\begin{definition}{\cite{Leroux13}}
\label{def:smooth}
Let \( \vect{P} \) be a periodic set.
\begin{itemize}
\item \(\vect{P}\) is \emph{asymptotically definable} if \(\dir(\vect{P})\) is a definable cone, i.e. \(\dir(\vect{P}) \setminus \{\mathbf{0}\}=\{\vect{x} \in \VectorSpace(\vect{P}) \mid A_1 \vect{x}>0, A_2 \vect{x}\geq 0\}\) for some integer matrices \(A_1, A_2\).

\item \(\vect{P}\) is \emph{well-directed} if every sequence \((\vect{p}_m)_{m \in \N}\) of vectors \(\vect{p}_m \in \vect{P}\) has an infinite subsequence \((\vect{p}_{m_k})_{k \in \N}\) such that \(\vect{p}_{m_k} - \vect{p}_{m_j} \in \dir(\vect{P})\) for all \(k \geq j\). 

\item \(\vect{P}\) is \emph{smooth} if it is asymptotically definable and well-directed.
\end{itemize}
\end{definition}

\begin{figure}[h!]
\begin{minipage}{4.5cm}
\begin{tikzpicture}
\begin{axis}[
    axis lines = left,
    xlabel = { },
    ylabel = { },
    xmin=0, xmax=8,
    ymin=0, ymax=8,
    xtick={0,2,4,6,8},
    ytick={0,2,4,6,8},
    ymajorgrids=true,
    xmajorgrids=true,
    thick,
    smooth,
    no markers,
]

\addplot[
    fill=blue,
    fill opacity=0.5,
    only marks,
    ]
    coordinates {
    (0,0)(1,1)(1,2)(1,3)(1,4)(1,5)(1,6)(1,7)(1,8)(2,1)(2,2)(2,3)(2,4)(2,5)(2,6)(2,7)(2,8)(3,1)(3,2)(3,3)(3,4)(3,5)(3,6)(3,7)(3,8)(4,1)(4,2)(4,3)(4,4)(4,5)(4,6)(4,7)(4,8)(5,1)(5,2)(5,3)(5,4)(5,5)(5,6)(5,7)(5,8)(6,1)(6,2)(6,3)(6,4)(6,5)(6,6)(6,7)(6,8)(7,1)(7,2)(7,3)(7,4)(7,5)(7,6)(7,7)(7,8)(8,1)(8,2)(8,3)(8,4)(8,5)(8,6)(8,7)(8,8)
    };

\end{axis}
\end{tikzpicture}
\end{minipage}%
\begin{minipage}{4.5cm}
\begin{tikzpicture}

\begin{axis}[
    axis lines = left,
    xlabel = { },
    ylabel = { },
    xmin=0, xmax=4,
    ymin=0, ymax=16,
    xtick={0,1,2,3,4},
    ytick={0,4,9,16},
    ymajorgrids=true,
    xmajorgrids=true,
    thick,
    smooth,
    no markers,
]
    
    
\addplot[
    fill=blue,
    fill opacity=0.5,
    only marks,
    ]
    coordinates {
    (0,0)(1,0)(1,1)(2,0)(2,1)(2,2)(2,3)(2,4)(3,0)(3,1)(3,2)(3,3)(3,4)(3,5)(3,6)(3,7)(3,8)(3,9)(4,0)(4,1)(4,2)(4,3)(4,4)(4,5)(4,6)(4,7)(4,8)(4,9)(4,10)(4,11)(4,12)(4,13)(4,14)(4,15)(4,16)
    };
    
\addplot[
    name path=A,
    domain=0:4,
    color=red,
]
{x^2};

\end{axis}
\end{tikzpicture}
\end{minipage}%
\begin{minipage}{4.5cm}
\begin{tikzpicture}
\begin{axis}[
    xlabel={ },
    ylabel={ },
    xmin=0, xmax=3,
    ymin=0, ymax=16,
    xtick={0,1,2,3},
    ytick={0,4,8,16},
    ymajorgrids=true,
    xmajorgrids=true,
    smooth,
    no markers,
]

    
\addplot+[
    name path=A,
    domain=0:3, 
    very thick,
    color=blue,
    ]
    {2*2^x};
    
\addplot+[
    name path=B,
    domain=0:3, 
    color=black,
    ]
    {0};
    
\addplot+[
    name path=CA,
    color=red,
    very thick,
    no marks,
    ]
    coordinates {
    (0,0)(3,10)
    };
    
\addplot+[
    name path=CB,
    color=red,
    very thick,
    no marks,
    ]
    coordinates {
    (1,0)(3,12.5)
    };
    
\addplot+[
    name path=CC,
    color=red,
    very thick,
    no marks,
    ]
    coordinates {
    (2,0)(3,14)
    };

\addplot[blue!40] fill between[of=CA and B];
\addplot[blue!40] fill between[of=CB and B];
\addplot[blue!40] fill between[of=CC and B];
    
\end{axis}
\end{tikzpicture}
\end{minipage}%

\caption{\textit{Left and middle}: The periodic sets \(\vect{P}=\{(0,0)\} \cup \N_{>0}^2 \) and \(\vect{P}=\{(x,y)\in \N^2 \mid y \leq x^2\}\) respectively.
Neither is finitely generated, but both are smooth with \(\Fill(\vect{P})=\N^2\). \newline
\textit{Right}: Underapproximation of \(\{(x,y)\mid y \leq 2^{x+1}\}\) via a union of three cones. The starting points are respectively \((0,0), (1,0)\) and \((2,0)\).}\label{FigureExamplesNonFinitelyGeneratedPeriodicSet}
\end{figure}
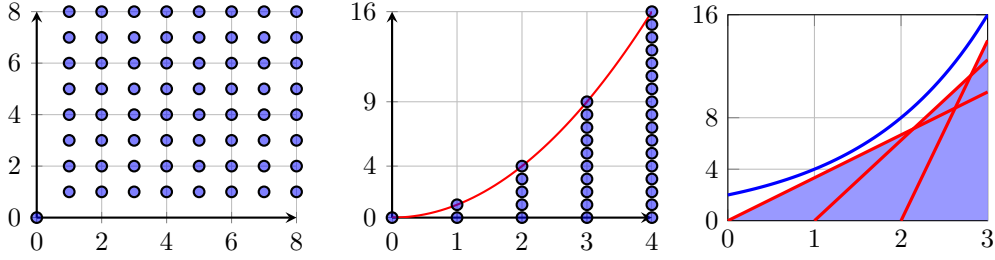

Figure \ref{FigureExamplesNonFinitelyGeneratedPeriodicSet} shows two examples of smooth periodic sets that are not finitely generated. 

\begin{example}
\label{ex:smooth}
Examples of non-smooth sets are \(\vect{P}_1=\{(x,y) \mid x \geq \sqrt{2} y\}\)  and \(\vect{P}_2=(\{(0,1)\} \cup \{(2^m,1) \mid m \in \N\})^{\ast}=\{(x,n) \in \N^2 \mid x \text{ has at most } n \text{ bits set to }1\text{ in the binary representation.}\}\). $\vect{P}_1$  is not asymptotically definable, because defining \(\dir(\vect{P})\) requires irrationals, while $\vect{P}_2$ is not well-directed (see observation 2 below). 
\end{example}

Intuitively, the ``boundaries'' of a smooth periodic set in two dimensions are either straight lines or function graphs ``curving outward'',  as in the example on the right of Figure \ref{FigureExamplesNonFinitelyGeneratedPeriodicSet}. 

We make a few observations:
\begin{enumerate}
\item The set \(\dir(\vect{P})\) is a cone. Indeed, if two lines in different directions \(\vect{d}\) and \(\vect{d}'\) are contained in \(\vect{P}\), then by periodicity \(\vect{P}\) also contains a \(\vect{d},\vect{d}'\) plane, and so \(\vect{P}\) contains a line in every direction between \(\vect{d}\) and \(\vect{d}'\). 
\item The most important case of Definition \ref{def:smooth} is when the \(\vect{p}_m\) are all on the same infinite line \(\vect{x}+\vect{d} \cdot \N\). Then the definition equivalently states that \(\vect{d} \in \dir(\vect{P})\), i.e. some infinite line in direction \(\vect{d}\) is contained in \(\vect{P}\). This makes sets where points are ``too scarce'' non-smooth. For instance, the set $\vect{P}_2$ of Example \ref{ex:smooth} contains infinitely many points on a horizontal line, but no full horizontal line, which would correspond to an arithmetic progression.
\end{enumerate}

\subsection{Fills of Smooth Sets are Finitely Generated}
\label{subsec:SmoothImpliesFinitelyGenerated}

We show that, while a smooth periodic set $\vect{P}$ may not be finitely generated, the set $\Fill(\vect{P})$ always is. We start with the following lemma.

\begin{restatable}{lemma}{LemmaDirectionsContainInterior}
Let \(\vect{P}\) be a periodic set. Then \(\interior(\overline{\mathbb{Q}_{\geq 0}\vect{P}}) \subseteq \mathbb{Q}_{\geq 0}\vect{P} \subseteq \dir(\vect{P}) \subseteq \overline{\mathbb{Q}_{\geq 0}\vect{P}}.\)

In particular, all these sets have the same closure. \label{LemmaDirectionsContainInterior}
\end{restatable}

\begin{proof}
Let \(\vect{x} \in \interior(\overline{\vect{C}})\), where \(\vect{C}:=\mathbb{Q}_{\geq 0}\vect{P}\). Then there exists \(\varepsilon >0\) such that the open ball \(B(\vect{x},\varepsilon)\) of radius \(\varepsilon\) around \(\vect{x}\) is contained in \(\overline{\vect{C}}\) by definition of interior. Hence for every \(\vect{y}\in B(\vect{x},\frac{\varepsilon}{2})\), there exists \(f(\vect{y}) \in B(\vect{y}, \frac{\varepsilon}{4}) \cap \vect{C}\) by definition of closure. We have surrounded \(\vect{x}\) by points \(f(\vect{y})\in \vect{C}\), hence by convexity of \(\vect{C}\) we have \(\vect{x}\in \vect{C}\).

Let \(\vect{d} \in \mathbb{Q}_{\geq 0}\vect{P}\). Then there exists \(m \in \N\) such that \(m \vect{d} \in \vect{P}\), in particular \(\N \cdot m \vect{d} \subseteq \vect{P}\).

Let \(\vect{d}\in \dir(\vect{P})\). Then by replacing \(\vect{d}\) by a multiple \(m \vect{d}\), there exists \(\vect{x}\) such that \(\vect{x}+ \N \cdot \vect{d} \subseteq \vect{P}\). We define the sequence \((\vect{x}_m)_{m \in \N}\) via \(\vect{x}_m:=\frac{1}{m}(\vect{x}+m \cdot \vect{d}) \in \mathbb{Q}_{\geq 0}\vect{P}\), and observe that its limit is \(\vect{d}\), i.e. \(\vect{d}\in \overline{\mathbb{Q}_{\geq 0}\vect{P}}\).
\end{proof}

\begin{example}
The set on the left of Figure \ref{FigureExamplesNonFinitelyGeneratedPeriodicSet} satisfies $\interior(\overline{\mathbb{Q}_{\geq 0}\vect{P}}) = \mathbb{Q}_{\geq 0}\vect{P} \subsetneq \dir(\vect{P}) = \overline{\mathbb{Q}_{\geq 0}\vect{P}}$. 
Indeed, $\interior(\overline{\mathbb{Q}_{\geq 0}\vect{P}})$ contains every direction except north and east, but they both belong to \(\dir(\vect{P})\). 
The middle set satisfies $\interior(\overline{\mathbb{Q}_{\geq 0}\vect{P}}) \subsetneq \mathbb{Q}_{\geq 0}\vect{P} = \dir(\vect{P}) \subsetneq \overline{\mathbb{Q}_{\geq 0}\vect{P}}$, since $\interior(\overline{\mathbb{Q}_{\geq 0}\vect{P}})$ contains neither north nor east, $\dir(\vect{P})$ contains east, and $\overline{\mathbb{Q}_{\geq 0}\vect{P}}$ contains both.
\end{example}

We are now ready to reprove the result:

\begin{proposition}{\cite[Lemma 5.1]{Leroux11}}
\label{prop:Fillisfinitelygenerated}
Let \(\vect{P}\) be smooth. Then \(\Fill(\vect{P})\) is full and hence f.g.
\end{proposition}

\begin{proof}
Since \(\vect{P}\) is smooth, \(\dir(\vect{P})\) is definable by definition. By Lemma \ref{LemmaDirectionsContainInterior} we have \(\overline{\mathbb{Q}_{\geq 0}\vect{P}}=\overline{\dir(\vect{P})}\). So $\overline{\mathbb{Q}_{\geq 0}\vect{P}}$ is the closure of a definable cone, and hence finitely generated by Lemma \ref{LemmaFinitelyGeneratedCones}. Hence \(\vect{P}=\Fill(\vect{P})\) is the intersection of a f.g.\ cone and a lattice, and hence full.
\end{proof}

\subsection{Underapproximating Periodic Sets} \label{SectionUnderapproximationPeriodicSets}

In Section \ref{subsec:SmoothImpliesFinitelyGenerated} we have seen that smooth periodic sets can be overapproximated by full linear sets in a natural way. Let us combine this with an underapproximation, mainly to provide a formal basis for the boundary function intuition above.

\begin{proposition} \cite[Lemma F.1]{Leroux13}
Let \(\vect{P}\) be a periodic set. Let \(\vect{F} \subseteq \mathbb{Q}^n\) finite. 

\(\vect{F} \subseteq (\vect{P}-\vect{P}) \cap \dir(\vect{P})\) if and only if there exists \(\vect{x}\) such that \(\vect{x}+\vect{F}^{\ast} \subseteq \vect{P}\). \label{PropositionAlmostPeriodicityPeriodicSet}
\end{proposition}

%
%

Now consider any finitely generated cone \(\vect{C} \subseteq \dir(\vect{P})\). Then \(\vect{C} \cap (\vect{P}-\vect{P})\) is full and hence finitely generated by some set \(\vect{F}\). By applying Proposition \ref{PropositionAlmostPeriodicityPeriodicSet}, we obtain a vector \(\vect{x}_{\vect{C}} \in \vect{P}\) such that \(\vect{x}_{\vect{C}} + (\vect{C} \cap (\vect{P}-\vect{P})) \subseteq \vect{P}\). This should be viewed as follows: Interpret the lattice \(\vect{P}-\vect{P}\) as the set of ``candidates'' for being in \(\vect{P}\). Namely, since \(\vect{x}_{\vect{C}} \in \vect{P}\), a vector \(\vect{x}_{\vect{C}}+ \vect{v}\) can only be in \(\vect{P}\) if \(\vect{v} \in \vect{P}-\vect{P}\). Then \(\vect{x}_{\vect{C}}+(\vect{C} \cap (\vect{P}-\vect{P})) \subseteq \vect{P}\) shows that every candidate in the given shifted cone (base point non-zero, so strictly speaking not a cone according to our definition) is actually in \(\vect{P}\). Repeating this process for larger and larger cones \(\vect{C}\), we obtain an underapproximation of \(\vect{P}\) of the form \(\bigcup_{f.g.\  \vect{C}} (\vect{x}_{\vect{C}} + \vect{C}) \cap (\vect{P}-\vect{P})\). The union of wider and wider shifted cones intuitively has a convex function as upper and a concave function as lower bound, as shown in the right of Figure \ref{FigureExamplesNonFinitelyGeneratedPeriodicSet}.

Observe that this lower bound did not use smoothness, in general this might hence be a strict underapproximation, as shown in the right of Figure \ref{FigureExamplesNonFinitelyGeneratedPeriodicSet}.

\subsection{Intersection of Smooth Sets} \label{SectionIntersectionPeriodicSets}
We would like smooth sets to be closed under intersection. Further, we would like that the fill of an intersection of smooth sets is the intersection of the fills. However, this does not hold in general. The following is a counterexample. 

\begin{example}
Define \(\vect{P}:=\{\vect{0}\} \cup \N_{>0}^2\), see left of Figure \ref{FigureExamplesNonFinitelyGeneratedPeriodicSet}, and \(\vect{P}'=\{(0,1)\}^{\ast}\), the \(y\)-axis. We have 
\(\{ \mathbf{0}\}=\dir(\vect{P} \cap \vect{P}') \subsetneq \dir(\vect{P}) \cap \dir(\vect{P}')\). Also, \(\{\vect{0}\}=\Fill(\vect{P} \cap \vect{P}') \subsetneq \Fill(\vect{P}) \cap \Fill(\vect{P}')=\vect{P}'\).
\end{example}

Fortunately, we can prove (see the Appendix): Smooth sets $\vect{P}, \vect{P}'$ such that $\Fill(\vect{P})$, $\Fill(\vect{P}')$, and $\Fill(\vect{P}) \cap \Fill(\vect{P}')$ have the same dimension behave well under intersection.



\begin{restatable}{proposition}{PropositionIntersectionPeriodicSets}
Let \(\vect{P},\vect{P}'\) be smooth periodic sets such that

\(\dim(\Fill(\vect{P}) \cap \Fill(\vect{P}'))=\dim(\Fill(\vect{P}))=\dim(\Fill(\vect{P}')). \text{ Then }\)

\begin{enumerate}
\item \(\dim(\vect{P} \cap \vect{P}')=\dim(\vect{P})=\dim(\vect{P}')\).
\item \(\dir(\vect{P} \cap \vect{P}')=\dir(\vect{P}) \cap \dir(\vect{P}')\).
\item \(\Fill(\vect{P} \cap \vect{P}')=\Fill(\vect{P}) \cap \Fill(\vect{P}')\).
\item \(\vect{P} \cap \vect{P}'\) is smooth.
\end{enumerate} \label{PropositionIntersectionPeriodicSets}
\end{restatable}

\section{Petri sets and Hybridizations} \label{Sec:NewPeriodicityProperty}


We introduce the remaining classes of sets used in our main result: Petri sets and sets admitting a hybridization.
Petri sets were introduced in \cite{Leroux11,Leroux12,Leroux13}. Hybridizations are a novel notion, and 
play a fundamental role in our main result. 

\subsection{Petri sets}

Leroux introduced almost semilinear sets and developed their theory in \cite{Leroux12,Leroux13}. Intuitively, 
they generalize semilinear sets by replacing linear sets with smooth periodic sets.

\begin{definition}{\cite{Leroux12,Leroux13}}
A set \(\vect{X}\) is \emph{almost linear} if \(\vect{X}=\vect{b}+\vect{P}\), where \(\vect{b}\in \mathbb{N}^n\) and \(\vect{P}\) is a smooth periodic set, and \emph{almost semilinear} if it is a finite union of almost linear sets.
\end{definition}


It was shown in \cite{Leroux12,Leroux13} that VAS reachability sets are almost semilinear. However,
it is easy to find almost semilinear sets that are not reachability sets of any VAS. Intuitively, the definition of a smooth periodic set only restricts the ``asymptotic behavior'' of the set, which can be ``simple'' even if the set itself is very ``complex''.

\begin{example}
Let \(\vect{Y} \subseteq \N_{>0}\) be any set. Then \(\vect{P}:=\{(0,0)\} \cup (\{1\} \times \vect{Y}) \cup \N_{>1}^2\) is a smooth periodic set;  indeed, \(\vect{P}\) contains a line in every direction, and is thus well-directed and asymptotically definable.
So \(\vect{P}\)  is almost semilinear. 
 \label{ExampleUglyPeriodicSet}
\end{example}

A way to eliminate at least some of these sets is to require that every intersection of the set with a semilinear set is still almost semilinear, a property enjoyed by all VAS reachability sets. For instance, assume that
in Example \ref{ExampleUglyPeriodicSet} the set \(\vect{Y}\) is not almost semilinear. Since the intersection of $\vect{P}$ and the linear set \((1,0)+(0,1) \cdot \N\) is equal to \(\vect{Y}\), we can eliminate $\vect{P}$. This idea leads to the notion of a Petri set.

\begin{definition}{\cite{Leroux12,Leroux13}}
A set \(\vect{X}\) is called a \emph{Petri set} if every intersection \(\vect{X}\cap \vect{S}\) with a semilinear set \(\vect{S}\) is almost semilinear. 
\end{definition}

All smooth periodic sets shown so far are also Petri sets. To see that the positive examples are indeed Petri sets we can use the following strong theorem from \cite{Leroux13}.

\begin{theorem}
\cite[Theorem IX.1]{Leroux13} Reachability sets of VAS are Petri sets. \label{TheoremVASPetriSets}
\end{theorem}

Many sets of the form $\{ (x, y) \mid  y \leq f(x)\}$ for convex \(f\), or $\{ (x, y) \mid  y \geq f(x)\}$ for concave \(f\), and boolean combinations thereof, are VAS reachability sets, and hence Petri sets.

\subsection{Hybridizations}


Given a Petri set $\vect{X} \subseteq \N^n$, it would be very useful to be able to partition $\N^n$ into finitely many semilinear regions $\vect{S}_1, \ldots, \vect{S}_k$ such that the sets $\vect{S_i} \cap \vect{X}$ have a simpler structure. 
In particular, we would like $\vect{S_i} \ \cap \vect{X}$ to be almost linear. Unfortunately, for some Petri sets no such partition exists
(example is in Appendix \ref{SectionNoAlmostLinearPartition}). We replace almost linearity by a slightly weaker notion for which the partition always exists: having a  hybridization
(Definition \ref{def:hybridization}). 



\begin{figure}[t]
\begin{minipage}{4.6cm}
\begin{tikzpicture}
\begin{axis}[
    xlabel={ },
    ylabel={ },
    xmin=0, xmax=4,
    ymin=0, ymax=16,
    xtick={0,1,2,3,4},
    ytick={0,4,9,16},
    ymajorgrids=true,
    xmajorgrids=true,
]

\addplot+[
    name path=AA,
    color=black,
    no marks,
    domain=0:4,
    ]
    {x^2+1};
    
\addplot+[
    name path=BA,
    color=black,
    no marks,
    ]
    coordinates {
    (0,1)(4,1)
    };
    
\addplot[
    color=blue!40,
]
fill between[of=AA and BA];

\addplot+[
    name path=AB,
    color=black,
    no marks,
    domain=1.5:4,
    ]
    {(x-1.5)^2+3};
    
\addplot+[
    name path=BB,
    color=black,
    no marks,
    ]
    coordinates {
    (1.5,3)(4,3)
    };
    
\addplot[
    color=red!40,
    domain=1.5:4,
]
fill between[of=AB and BB];

%
%
    
\end{axis}
\end{tikzpicture}
\end{minipage}%
\begin{minipage}{4.6cm}
\begin{tikzpicture}
\begin{axis}[
    xlabel={ },
    ylabel={ },
    xmin=0, xmax=12,
    ymin=0, ymax=12,
    xtick={0,2,4,6,8,10,12},
    ytick={0,2,4,6,8,10,12},
    ymajorgrids=true,
    xmajorgrids=true,
]

\addplot[
    fill=blue,
    fill opacity=0.5,
    only marks,
    ]
    coordinates {
    (0,0)(0,3)(0,4)(0,5)(0,6)(0,7)(0,8)(0,9)(0,10)(0,11)(0,12)(1,4)(1,5)(1,6)(1,7)(1,8)(1,9)(1,10)(1,11)(1,12)(2,5)(2,6)(2,7)(2,8)(2,9)(2,10)(2,11)(2,12)(3,5)(3,6)(3,7)(3,8)(3,9)(3,10)(3,11)(3,12)(4,6)(4,7)(4,8)(4,9)(4,10)(4,11)(4,12)(5,6)(5,7)(5,8)(5,9)(5,10)(5,11)(5,12)(6,6)(6,7)(6,8)(6,9)(6,10)(6,11)(6,12)(7,6)(7,7)(7,8)(7,9)(7,10)(7,11)(7,12)(8,7)(8,8)(8,9)(8,10)(8,11)(8,12)(9,7)(9,8)(9,9)(9,10)(9,11)(9,12)(10,7)(10,8)(10,9)(10,10)(10,11)(10,12)(11,7)(11,8)(11,9)(11,10)(11,11)(11,12)(12,7)(12,8)(12,9)(12,10)(12,11)(12,12)
    };
    
\addplot[
    fill=green,
    fill opacity=0.5,
    only marks,
    ]
    coordinates {
    (1,0)(1,1)(2,0)(2,1)(2,2)(2,3)(2,4)(3,0)(3,1)(3,2)(3,3)(3,4)(4,0)(4,1)(4,2)(4,3)(4,4)(4,5)(5,0)(5,1)(5,2)(5,3)(5,4)(5,5)(6,0)(6,1)(6,2)(6,3)(6,4)(6,5)(7,0)(7,1)(7,2)(7,3)(7,4)(7,5)(8,0)(8,1)(8,2)(8,3)(8,4)(8,5)(8,6)(9,0)(9,1)(9,2)(9,3)(9,4)(9,5)(9,6)(10,0)(10,1)(10,2)(10,3)(10,4)(10,5)(10,6)(11,0)(11,1)(11,2)(11,3)(11,4)(11,5)(11,6)(12,0)(12,1)(12,2)(12,3)(12,4)(12,5)(12,6)
    };
    
\addplot[
    domain=0:12,
    color=red,
]
{log2(x+1)+3};

\addplot[
    domain=0:12,
    color=red,
]
{x^2};
    
\end{axis}
\end{tikzpicture}
\end{minipage}
\begin{minipage}{4.6cm}
\begin{tikzpicture}
\begin{axis}[
    xlabel={ },
    ylabel={ },
    xmin=0, xmax=12,
    ymin=0, ymax=12,
    xtick={0,2,4,6,8,10,12},
    ytick={0,2,4,6,8,10,12},
    ymajorgrids=true,
    xmajorgrids=true,
]

\addplot[
    fill=blue,
    fill opacity=0.5,
    only marks,
    ]
    coordinates {
    (1,1)(2,2)(3,2)(3,3)(4,3)(4,4)(5,3)(5,4)(5,5)(6,3)(6,4)(6,5)(6,6)(7,3)(7,4)(7,5)(7,6)(7,7)(8,4)(8,5)(8,6)(8,7)(8,8)(9,4)(9,5)(9,6)(9,7)(9,8)(9,9)(10,4)(10,5)(10,6)(10,7)(10,8)(10,9)(10,10)(11,4)(11,5)(11,6)(11,7)(11,8)(11,9)(11,10)(11,11)(12,4)(12,5)(12,6)(12,7)(12,8)(12,9)(12,10)(12,11)(12,12)
    };
    
\addplot[
    fill=green,
    fill opacity=0.5,
    only marks,
    ]
    coordinates {
    (0,0)(1,0)(2,0)(3,0)(4,0)(5,0)(6,0)(7,0)(8,0)(9,0)(10,0)(11,0)(12,0)
    };
    
\addplot[
    color=red,
    ]
    coordinates {
    (0,0)(12,12)
    };
    
\end{axis}
\end{tikzpicture}
\end{minipage}

\caption{
\textit{Left}: An almost linear set $\vect{X} = \vect{b} + \vect{P}$ with \(\vect{b}=(0,1)\) and \(\vect{P}=\{(x,y) \mid y \leq x^2\}\) (in blue). The property \(\vect{X}+\vect{P}\subseteq \vect{X}\) implies that  the ``translation'' of $\vect{X}$ to \emph{any} point in the set (shown in brown for a particular point) is included in the set. \label{FigureIntuitionAlmostHybridlinear} \newline
\textit{Middle}: The two smooth periodic sets \(\vect{P}_1:=\{(x,y) \mid y \geq \log_2(x+1)+3\} \cup \{(0,0)\}\) in blue and \(\vect{P}_2:=\{(x,y) \mid y \leq x^2\}\) in green. Their union is almost hybridlinear, but not almost linear. \newline
\textit{Right}: The smooth periodic sets \(\vect{P}_1:=\{(x,y)\mid x \geq y \geq \log_2(x+1)\}\) and \(\vect{P}_2:=\{(1,0)\}^{\ast}\). The union \(\vect{X}\) does not have a hybridization, since \(\vect{P}=\{(0,0)\}\) is the only possibility to fulfill \(\vect{X}+\vect{P} \subseteq \vect{X}\).}
\end{figure}

A set is almost linear if there exists a vector $\vect{b}$ and a smooth periodic set $\vect{P}$ such that $\vect{X} = \vect{b} + \vect{P}$. The following definition is equivalent: There exists a vector $\vect{b}$ and a smooth periodic set $\vect{P}$ such that \(\vect{b} \in \vect{X}\) and \(\vect{X}+\vect{P} \subseteq \vect{X} \subseteq \vect{b} + \vect{P}\).

We weaken this condition by requiring only the existence of a vector $\vect{b}$ and a smooth periodic set $\vect{P}$ such that \(\vect{X}+\vect{P} \subseteq \vect{X} \subseteq \vect{b} + \Fill(\vect{P})\).

That is, we drop the condition $\vect{b} \in \vect{X}$, and replace $\vect{P}$ on the right by the possibly larger set $\Fill(\vect{P})$. (For example, the periodic sets on the left of Figure \ref{FigureIntuitionAlmostHybridlinear} as well as in the middle satisfy \(\Fill(\vect{P})=\N^2\)). We then call the set $\vect{b} + \Fill(\vect{P})$ a hybridization of $\vect{X}$. The formal definition is as follows, where for technical reasons we also introduce weak hybridizations. 

\begin{definition}
\label{def:hybridization}
Let \(\vect{X} \subseteq \mathbb{N}^n\) be non-empty. A set $\vect{H}$ is a \emph{weak hybridization} of $\vect{X}$ if there exists a finite set $\vect{B} \subseteq \N^n$ and a smooth periodic set $\vect{P}$ such that \(\vect{H}=\vect{B}+ \Fill(\vect{P})\) and \(\vect{X} + \vect{P} \subseteq \vect{X} \subseteq \vect{H}$. If \(\vect{B}=\{\vect{b}\}\), then \(\vect{H}\) is a \emph{hybridization} of $\vect{X}$. 
\end{definition}

\begin{remark}
A full linear weak hybridization does not guarantee existence of a hybridization. For example \(\vect{X}=1+3 \N \cup 2+3\N\) has weak hybridization \(\vect{H}=\{0,1,2\}+3\N = \N\). However, since \(\vect{X}\) does not contain any points congruent to \(0\) modulo \(3\), any periodic set \(\vect{P}\) fulfilling \(\vect{X}+\vect{P} \subseteq \vect{X}\) has to fulfill \(\vect{P} \subseteq 3\N\). Hence \(\vect{B}\) cannot be chosen as a singleton. \label{RemarkHybridizationRepresentationDependent}
\end{remark}

It follows from this definition that almost linear sets have hybridizations. The reason for the name (weak) hybridization is that the set $\vect{H}$ is always hybridlinear, a notion introduced in \cite{ginsburg1963bounded} by Ginsburg and Spanier and later studied in \cite{ChistikovH16} by Chistikov and Haase. We recall the definition for future reference.

\begin{definition}
A set \(\vect{H} \subseteq \N^n\) is \emph{hybridlinear} if \(\vect{H}= \vect{B} + \vect{P}\) for some finite set \(\vect{B}\) and some finitely generated periodic set \(\vect{P}\subseteq \N^n\).
\end{definition}

We end this section with a characterization of the sets that admit weak hybridizations. 

\begin{definition}
\label{def:almosthybridlinear}
A non-empty set \(\vect{X} \subseteq \N^n\) is \emph{almost hybridlinear} if there exist \(\vect{b}_1, \dots, \vect{b}_r \in \mathbb{N}^n\) and smooth \(\vect{P}_1, \dots, \vect{P}_r\) with \(\vect{X}=\bigcup_{i=1}^r \vect{b}_i + \vect{P}_i\), such that \(\Fill(\vect{P}_i)=\Fill(\vect{P}_j)\) for all \(i,j\). 
\end{definition}


\begin{restatable}{theorem}{TheoremEquivalentAlmostHybridlinearCondition}
A non-empty Petri set \(\vect{X} \subseteq \N^n\) is almost hybridlinear if and only if it has a weak hybridization. 
\label{TheoremEquivalentAlmostHybridlinearCondition}
\end{restatable}

This theorem helps  to find examples of non-trivial hybridizations (i.e. not of type \(\vect{P}\) has hybridization \(\Fill(\vect{P})\)).  For example \([(0,1)+\vect{P}_1] \cup [(0,6)+\vect{P}_2]\) for \(\vect{P}_1=\{(x,y) \in \N^2 \mid y \leq x^2\}\) and \(\vect{P}_2=\{(x,y) \in \N^2 \mid y \geq \log_2(x+1)\}\) has weak hybridization \(\N^2\), since \(\Fill(\vect{P}_1)=\Fill(\vect{P}_2)=\N^2\). This is very similar to the middle of Figure \ref{FigureIntuitionAlmostHybridlinear}. On the other hand, in the right of Figure \ref{FigureIntuitionAlmostHybridlinear} the smooth periodic sets barely intersect, and then the union is usually not almost hybridlinear.

\section{Proof of Theorem \ref{TheoremFinalPartition}} \label{Sec:PetriSetsAndDimension}


In this section we prove Theorem \ref{TheoremFinalPartition}. The algorithm and its proof will refine the partition in three steps, respectively described in Section \ref{subsec:ArbitraryPartition}, Section \ref{subsec:FullLinearPartition} and Section \ref{SectionReducibility}: During the first two steps the sets $\vect{X} \cap \vect{S}_i$ are not required to be irreducible, and in addition after the first step, the \(\vect{S}_i\) are allowed to be weak hybridizations instead of hybridizations.

\subsection{Existence of a Hybridlinear Partition}
\label{subsec:ArbitraryPartition}

We collect five important properties of (weak) hybridizations in Proposition \ref{PropositionPropertiesOfHybridization}. Then, we use these properties to formulate a procedure for producing a partition \(\vect{S}=\vect{S}_1 \cup \dots \cup \vect{S}_k\) of sets, not necessarily full linear, satisfying the properties of Theorem \ref{TheoremFinalPartition} except for irreducibility. The procedure is described in Figure \ref{fig:Partition}. It is effective for VAS reachability sets, but not in general. 

We start by reminding that the class of hybridlinear sets is closed under intersection. 

\begin{restatable}{lemma}{LemmaIntersectionLinearSets}
\cite[Lemma 7.8]{Leroux09} Let \(\vect{b}_1+\vect{Q}_1\) and \(\vect{b}_2+\vect{Q}_2\) be linear sets. Then \((\vect{b}_1+\vect{Q}_1) \cap (\vect{b}_2+\vect{Q}_2)=\vect{B}+(\vect{Q}_1 \cap \vect{Q}_2)\) for some finite \(\vect{B}\). \label{LemmaIntersectionLinearSets}
\end{restatable}

\begin{proposition}\label{PropositionPropertiesOfHybridization}
The following statements hold:
\begin{enumerate}
\item[1)] If \(\vect{H}\) is a weak hybridization of \(\vect{X}\), then \(\dim(\vect{X})=\dim(\vect{H})\).
\item[2)] If \(\vect{H}\) is a weak hybridization of \(\vect{X}\) and \(\vect{L}=\vect{b}+\vect{Q}\) full linear s.t. \(\dim(\vect{H} \cap \vect{L})=\dim(\vect{H})=\dim(\vect{L})\), then \(\vect{H} \cap \vect{L}\) is a weak hybridization for \(\vect{X} \cap \vect{L}\), or \(\vect{X} \cap \vect{L}\) is empty.
\item[3)] If \(\vect{H}\) is a (weak) hybridization for both \(\vect{X}_1\) and \(\vect{X}_2\), then \(\vect{H}\) is a (weak) hybridization for \(\vect{X}_1 \cup \vect{X}_2\).
\item[4)] For every Petri set \(\vect{X}\) and semilinear \(\vect{S}\) there is a partition \(\vect{X} \cap \vect{S}=\vect{X}_1 \cup \dots \cup \vect{X}_r\) of $\vect{X} \cap \vect{S}$ such that every \(\vect{X}_i\) has a (true) hybridization \(\vect{L}_i\).
\item[5)] If $\vect{X}$ is the reachability set of a VAS, then the set $\{\vect{L}_1, \ldots, \vect{L}_r\}$ of hybridizations of part 4) is computable.
\end{enumerate}
\end{proposition}

\begin{proof}
For proofs 1) and 2), write \(\vect{H}:=\vect{B}+\Fill(\vect{P})\), where \(\vect{P}\) is smooth and \(\vect{X}+\vect{P} \subseteq \vect{X}\).

\smallskip\noindent 1): This follows from the properties of dimension in Lemmas \ref{BasicDimensionProperties} and \ref{LemmaFromJerome}. In particular, \(\dim(\vect{P})=\dim(\vect{V})\), where \(\vect{V}\) is the vector space generated by \(\vect{P}\), also implies \(\dim(\vect{P})=\dim(\Fill(\vect{P}))\). Hence \(\vect{X} \subseteq \vect{H}\) implies \(\dim(\vect{X}) \leq \dim(\vect{P})\). Since \(\vect{X}\) is non-empty, \(\vect{X}+\vect{P} \subseteq \vect{X}\) implies \(\dim(\vect{X}) \geq \dim(\vect{P})\).

\smallskip\noindent 2): By Lemma \ref{LemmaIntersectionLinearSets}, \(\vect{H} \cap \vect{L}=\vect{F}+(\Fill(\vect{P}) \cap \vect{Q})\) for some finite set \(\vect{F}\). By Proposition \ref{PropositionIntersectionPeriodicSets}, we have that \(\vect{P} \cap \vect{Q}\) is smooth and \(\Fill(\vect{P} \cap \vect{Q})=\Fill(\vect{P}) \cap \Fill(\vect{Q})=\Fill(\vect{P}) \cap \vect{Q}\). We have \(\vect{X} \cap \vect{L} \subseteq \vect{H} \cap \vect{L}\). We also have \((\vect{X} \cap \vect{L})+(\vect{P} \cap \vect{Q}) \subseteq \vect{X} + \vect{P} \subseteq \vect{X}\) and \((\vect{X} \cap \vect{L})+ (\vect{P} \cap \vect{Q}) \subseteq \vect{L}+\vect{Q} \subseteq \vect{L}\), hence \(\vect{H} \cap \vect{L}\) is a weak hybridization of \(\vect{X} \cap \vect{L}\).

\smallskip\noindent 3): Write \(\vect{B}_1+\Fill(\vect{P}_1)=\vect{H}=\vect{B}_2+\Fill(\vect{P}_2)\), where \(\vect{P}_1\) for \(\vect{X}_1\) and \(\vect{P}_2\) for \(\vect{X}_2\) are as in the definition of weak hybridization. By Lemma \ref{LemmaIntersectionLinearSets}, we have \(\vect{H}=\vect{H} \cap \vect{H}=\vect{F}+[\Fill(\vect{P}_1) \cap \Fill(\vect{P}_2)]\) for some finite set \(\vect{F}\). Define \(\vect{P}:=\vect{P}_1 \cap \vect{P}_2\) and \(\vect{X}:=\vect{X}_1 \cup \vect{X}_2\). By Proposition \ref{PropositionIntersectionPeriodicSets}, \(\vect{P}\) is smooth and \(\Fill(\vect{P})=\Fill(\vect{P}_1) \cap \Fill(\vect{P}_2)\). We also have \(\vect{X}+\vect{P} \subseteq \vect{X}\).

\smallskip\noindent 4): Since \(\vect{X}\) is a Petri set, \(\vect{X} \cap \vect{S}\) is almost semilinear, and can hence be written as \(\vect{X}=\bigcup_{i=1}^r \vect{b}_i + \vect{P}_i\) for smooth periodic sets \(\vect{P}_i \subseteq \N^n\) and points \(\vect{b}_i \in \N^n\). Every \(\vect{X}_i:=\vect{b}_i+\vect{P}_i\) is by definition almost hybridlinear with hybridization \(\vect{b}_i+\Fill(\vect{P}_i)\), which is a full linear set.

\smallskip\noindent 5): 4) can be computed using the Kosaraju-Lambert-Mayr-Sacerdote-Tenney (KLMST) decomposition \cite{Kosaraju82, Lambert92, Leroux09, LerouxS19}. The KLMST decomposition constructs a finite set of VASS-like objects, called perfect marked graph transition sequences or perfect MGTSs, such that the set of reachable configurations of the VAS is the union of the sets of reachable configurations of the perfect MGTSs. Further, for every perfect MGTS one can effectively construct a set of linear equations satisfying the following property: the set of solutions of the equation system is a hybridization of the set of reachable configurations of the perfect MGTS. The set of solutions of a system of linear equations is always hybridlinear. Moreover, for the systems derived from MGTSs one can show that the set has a full linear hybridization (e.g. \cite[Lemma 5.1]{Leroux09}). This gives us the desired hybridizations \(\vect{L}_1, \dots, \vect{L}_r\). \footnote{While Hauschildt already used the KLMST decomposition in \cite{Hauschildt90} in 1990, it took until 2019 \cite{LerouxS15, LerouxS19} to fully understand the theoretical aspects behind the algorithm and its complexity of Ackermann.}
\end{proof}

\begin{figure}[t]
\rule{\textwidth}{0.1cm}
Partition$(\vect{X}, \vect{S})$. Input: Petri set \(\vect{X}\) and semilinear set \(\vect{S}\):

\medskip

1) If \(\vect{S}\) is empty, return \(\vect{S}\). If $\vect{S}$ is not full, compute a partition $\vect{S}_1, \ldots, \vect{S}_r$ of $\vect{S}$ into full linear sets, return $\bigcup_{i=1}^r  \text{Partition}(\vect{X}, \vect{S}_i)$ and stop.

\smallskip
 
Otherwise, compute the set $\mathcal{L} =\{\vect{L}_1, \ldots, \vect{L}_r\}$ of hybridizations of the partition \(\vect{X}_1 \cup \dots \cup \vect{X}_r\) of $\vect{X} \cap \vect{S}$ given by Proposition \ref{PropositionPropertiesOfHybridization}(4), and move to step 2).
 
\smallskip
 
\emph{Remark: This step is not effective for arbitrary Petri sets, but it is effective for VAS reachability sets by Proposition \ref{PropositionPropertiesOfHybridization}(5).}

\smallskip

If $r=0$, i.e., if $\vect{X} \cap \vect{S}$ is empty, then return $\vect{S}$ and stop. Otherwise, move to step 2).

\medskip

2) For every $\vect{L}_i\in \mathcal{L}$ compute a decomposition $\mathcal{K}_i$ of $\vect{L}_i^C \cap \vect{S}$ into full linear sets, where $\vect{L}_i^C$ is the complement of $\vect{L}_i$, and move to step 3).

\medskip

3) Let $\mathcal{M}$ be the set of tuples  
$(\vect{M}_1, \ldots, \vect{M}_r) \in \left( \{ \vect{L}_1\} \cup \mathcal{K}_1 \right)  \times \cdots \times \left( \{ \vect{L}_r\} \cup \mathcal{K}_r \right).$

For every $M \in \mathcal{M}$, let $\vect{S}_{M} := \vect{S} \cap \vect{M}_1 \cap \cdots \cap \vect{M}_r$.
 
 \smallskip
 
\emph{Remark: $\{\vect{S}_{M} \mid M \in \mathcal{M} \}$ is a partition of $\vect{S}$.}

 \smallskip
 
For every $M \in \mathcal{M}$, define $P_{M}$ as follows: If $\dim(\vect{S}_{M}) < \dim(\vect{S})$, then $P_{M} := \text{Partition}(\vect{X}, \vect{S}_{M})$, otherwise $P_{M} := \{\vect{S}_{M}\}$. Output $\bigcup_{M \in \mathcal{M}} P_{M}$. 

\rule{\textwidth}{0.1cm}
\caption{The procedure Partition$(\vect{X}, \vect{S})$.}
\label{fig:Partition}
\end{figure}

\begin{proposition}
Let $\vect{X}$ be a Petri set and let $\vect{S}$ be a semilinear set. 
Partition$(\vect{X}, \vect{S})$ produces a partition \(\vect{S}=\vect{S}_1 \cup \dots \cup \vect{S}_k\) into pairwise disjoint hybridlinear sets (not necessarily full linear) such that for every \(i\) the set \(\vect{X} \cap \vect{S}_i\) is either empty or has weak hybridization \(\vect{S}_i\). Further, if $\vect{X}$ is the reachability set of a VAS, then the partition is computable. \label{PropositionAlmostHybridlinearPartition}
\end{proposition}

\begin{proof}
The procedure is depicted in Figure \ref{fig:Partition}, in addition we give an intuitive description of it: In Step 1) we first partition \(\vect{S}\) into full linear sets and consider them separately. So assume that $\vect{S}$ is a full linear set. The procedure uses Proposition \ref{PropositionPropertiesOfHybridization}(5) to compute a set of full linear hybridizations $\vect{L}_1, \ldots, \vect{L}_r$ of a partition $\vect{X}_1 \cup \cdots \cup \vect{X}_r$ of $\vect{X} \cap \vect{S}$. Step 2) considers all possible sets obtained by picking for each $i \in \{1, \ldots, r\}$ either the set $\vect{L}_i$ or a linear set of its complement (its complement is semilinear, and so a finite union of linear sets), and intersecting all of them. The procedure adds all the sets having full dimension to the output partition, and does a recursive call on the others.

Every step can be performed: The set $\mathcal{L}$ of Step 1 exists by Proposition \ref{PropositionPropertiesOfHybridization}(4). To check the dimension of a semilinear set \(\vect{S}=\bigcup_{j=1}^r \vect{b}_j+\vect{F}_j^{\ast}\), which is needed in step 3), we use Lemma \ref{LemmaFromJerome} to obtain that for \(\vect{F}_j^{\ast}\) this is simply the rank of the generator matrix, and by Lemma \ref{BasicDimensionProperties} we have \(\dim(\vect{S})=\max_j \dim(\vect{F}_j^{\ast})\).

\smallskip\noindent \textit{Termination:} Partition(\(\vect{X},\vect{S}\)) only performs a recursive call if \(\vect{S}\) is not a full linear set or on semilinear sets \(\vect{S}'\) with \(\dim(\vect{S}')<\dim(\vect{S})\), hence recursion depth is at most \(2\dim(\vect{S})+1\) and termination immediate.

\smallskip\noindent \textit{Correctness:} The proof obligation for correctness is that for every \(M=(\vect{M}_1, \dots, \vect{M}_r) \in \mathcal{M}\), where \(S_M\) fulfills \(\dim(\vect{S}_M)=\dim(\vect{S})\), \(\vect{X} \cap \vect{S}_M\) is either empty or has \(\vect{S}_M\) as weak hybridization. Therefore fix such \(M\). 

\begin{claim}
\(\dim(\vect{M}_j)=\dim(\vect{S})\) for all \(j\).
\end{claim}

\begin{proof}[Proof of Claim]\(\geq \dim(\vect{S})\) follows since all these sets contain \(\vect{S}_M\), which fulfills \(\dim(\vect{S}_M) = \dim(\vect{S})\). For the other direction, to prove ``\(\leq\)'' for \(j\) where we choose \(\vect{L}_j\) we have \(\dim(\vect{S}) \geq \dim(\vect{X} \cap \vect{S})= \max_j \dim(\vect{L}_j)\) by Proposition \ref{PropositionPropertiesOfHybridization}. For other \(j\) we use \(\vect{L}_j^C \cap \vect{S} \subseteq \vect{S}\).
\end{proof}

The claim allows us to use Proposition \ref{PropositionPropertiesOfHybridization}(2). Let \(\vect{X}_j\) be such that \(\vect{X} \cap \vect{S}=\bigcup_{j=1}^r \vect{X}_j\) and \(\vect{X}_j\) has hybridization \(\vect{L}_j\). By applying Proposition \ref{PropositionPropertiesOfHybridization}(2) enough times, for every \(j\) with \(\vect{M}_j=\vect{L}_j\), we obtain that \(\vect{X}_j \cap \vect{S}_M\) has weak hybridization \(\vect{S}_M\). This does not depend on \(j\) because intersecting with \(\vect{L}_j\) twice does not change the set. For all other \(j\) we have \(\vect{X}_j \cap \vect{S}_M=\emptyset\), since we intersect with the complement of an overapproximation. Hence \(\vect{X} \cap \vect{S}_M=\bigcup_{j, \vect{M}_j=\vect{L}_j} (\vect{X}_j \cap \vect{S}_M)\) has weak hybridization \(\vect{S}_M\) by Proposition \ref{PropositionPropertiesOfHybridization}(3), or is empty if we never chose \(\vect{M}_j=\vect{L}_j\).
\end{proof}

\subsection{Existence of a Full Linear Partition} 
\label{subsec:FullLinearPartition}

We show that Proposition \ref{PropositionAlmostHybridlinearPartition} can be strengthened to make the sets $\vect{S}_i$ not only hybridlinear, but even full linear, in a way that the sets \(\vect{S}_i\) are actually (true) hybridizations.

\begin{proposition}
Let \(\vect{X}\) be a Petri set. For every semilinear set \(\vect{S}\) there exists a partition \(\vect{S}=\vect{S}_1 \cup \dots \cup \vect{S}_k\) of $\vect{S}$ into pairwise disjoint full linear sets such that for every \(i\) the set \(X \cap \vect{S}_i\) is either empty or has hybridization \(\vect{S}_i\). Further, if $\vect{X}$ is the reachability set of a VAS, then the partition is computable. \label{PropositionFullLinearAlmostHybridlinearPartition}
\end{proposition}

\begin{proof} 
The main algorithm uses a subroutine with the same inputs and outputs as itself, but with the promise that \(\vect{X} \cap \vect{S}\) has weak hybridization \(\vect{S}\). We first describe the main algorithm, and then the subroutine.

Main algorithm: First apply Proposition \ref{PropositionAlmostHybridlinearPartition} to obtain a partition 
$\vect{S} = \vect{S}_1 \cup \cdots \cup \vect{S}_k$ into hybridlinear sets otherwise satisfying the conditions. Output \(\bigcup_{i=1}^k \text{Subroutine}(\vect{X}, \vect{S}_i)\).

Subroutine: If \(\vect{S}\) is already full linear, return \(\vect{S}\). Otherwise write \(\vect{S}=\{\vect{c}_1, \dots, \vect{c}_r\}+\Fill(\vect{P})\). Let \(j \sim k \iff \vect{c}_j-\vect{c}_k \in \Fill(\vect{P})-\Fill(\vect{P})=\vect{P}-\vect{P}\). Compute a system \(R\) of representatives for \(\sim\). For every \(i \in R\), define \(\vect{S}_i:=\vect{c}_i+\Fill(\vect{P})\). Define \(\vect{S}':=\vect{S} \setminus \bigcup_{i \in R} \vect{S}_i\) and output \(\{\vect{S}_i \mid i \in R\} \cup \text{MainAlgorithm}(\vect{X}, \vect{S}')\).

Termination: We prove that recursion depth \(\leq 2\dim(\vect{S})+1\) by proving \(\dim(\vect{S}') < \dim(\vect{S})\) in the subroutine. For every equivalence class \(C\) of \(\sim\), there exists \(\vect{c} \in \Z^n\) such that \(\vect{c_j}-\vect{c} \in \vect{P}\) for all \(j \in C\). To see this, fix some \(i \in C\), and write \(\vect{c}_j-\vect{c}_i=\vect{p}_j-\vect{p}_j' \in \vect{P}-\vect{P}\). Choose \(\vect{c}:=\vect{c}_i - \sum_{j \in C} \vect{p}_j'\).

Then \(\bigcup_{j \in C} \vect{c}_j+\Fill(\vect{P}) \subseteq \vect{c}+\Fill(\vect{P})\), and hence using Lemma \ref{LemmaRemovingInnerConeReducesDimension} we obtain 

\(\dim(\bigcup_{j \in C} \vect{c}_j+\Fill(\vect{P}) \setminus \vect{S}_i)\leq \dim(\vect{c}+\Fill(\vect{P}) \setminus \vect{c}_i+\Fill(\vect{P}))< \dim(\Fill(\vect{P})).\)

Correctness: The main algorithm is clearly correct if the subroutine is. In the subroutine, we have \(\vect{S}_i \cap \vect{S}_j=\emptyset\) since \(i \not \sim j\) for \(i,j \in R\). All \(\vect{S}_i\) are full linear by definition. Furthermore, \(\vect{X} \cap \vect{S}_i\) has weak hybridization \(\vect{H} \cap \vect{S}_i=\vect{S}_i\) by Proposition \ref{PropositionPropertiesOfHybridization}(2). To obtain that the hybridization is not weak, observe that Proposition \ref{PropositionPropertiesOfHybridization}(2) specifically shows that the intersection of the representations, which is the full linear representation of \(\vect{S}_i\), is a weak hybridization.
\end{proof}

\subsection{Reducibility of almost hybridlinear Sets} \label{SectionReducibility}

The final ingredient of our main result is reducibility. We name it after its counterpart in Hauschildt's PhD thesis \cite{Hauschildt90}.

\begin{definition}
A set \(\vect{X}\) with hybridization \(\vect{c}+\Fill(\vect{P})\) is \emph{reducible} if there exists \(\vect{x}\) such that \(\vect{x}+\Fill(\vect{P}) \subseteq \vect{X}\).
\end{definition}

In other words, \(\vect{X}\) is reducible if every large enough point of its hybridization is already in \(\vect{X}\). Observe that this does not follow from hybridization, as \(\Fill(\vect{P})\) is larger than \(\vect{P}\). Our usual examples of sets with hybridization are smooth periodic sets, these also illustrate reducibility: The set in the left of Figure \ref{FigureExamplesNonFinitelyGeneratedPeriodicSet} is reducible, while the middle is not. Another example of hybridization was in the middle of Figure \ref{FigureIntuitionAlmostHybridlinear}, this set is also reducible. In fact, whenever \(\vect{X}=\vect{b}+\vect{P}\), \(\vect{X}\) is reducible if and only if \(\dir(\vect{P})=\overline{\Q_{\geq 0}\vect{P}}\). Namely, use Proposition \ref{PropositionAlmostPeriodicityPeriodicSet} with \(\vect{F}\) the generators of \(\Fill(\vect{P})\). For other sets \(\vect{X}\), write \(\vect{X}=\bigcup_{i=1}^r \vect{b}_i + \vect{P}_i\) as almost hybridlinear set. Whether it is reducible again only depends on the cones \(\dir(\vect{P}_i)\), for a proof see the appendix. Since matrices for the definable cones \(\dir(\vect{P}_i)\) can in the case of VAS be determined using KLMST-decomposition \cite{Hauschildt90}, we obtain the following.

\begin{restatable}{theorem}{TheoremReducibilityIsDecidable}
\cite[even without promise]{Hauschildt90} The following problem is decidable.

Input: Reachability set \(\vect{R}\), represented via the transitions of the VASS, full linear set \(\vect{S}\). 

Promise: \(\vect{R} \cap \vect{S}\) has hybridization \(\vect{S}\).

Output: Is \(\vect{R} \cap \vect{S}\) reducible? \label{TheoremReducibilityIsDecidable}
\end{restatable}

We can now prove our main result. 

\TheoremFinalPartition*

\begin{proof}

Step 1: Use Proposition \ref{PropositionFullLinearAlmostHybridlinearPartition} to compute a partition \(\vect{S}=\vect{S}_1 \cup \dots \cup \vect{S}_k\) into full linear sets such that \(\vect{X} \cap \vect{S}_i\) has hybridization \(\vect{S}_i\) if it is non-empty. For every set \(\vect{S}_i\) with \(\vect{X} \cap \vect{S}_i \neq \emptyset\) do Step 2.

Step 2: Decide whether \(\vect{X} \cap \vect{S}_i\) is reducible using Theorem \ref{TheoremReducibilityIsDecidable}. If irreducible, output \(\vect{S}_i\). Otherwise, there exists \(\vect{x}\) such that \(\vect{x}+\vect{Q} \subseteq \vect{X} \cap \vect{S}_i\), where \(\vect{S}_i=\vect{c}+\vect{Q}\). Find such an \(\vect{x}\), add \(\vect{x}+\vect{Q} \subseteq \vect{X}\) to the final partition and do a recursive call on \(\vect{S}_i \setminus (\vect{x}+\vect{Q})\).

Termination: We claim that we only perform recursion on \(\vect{S}'\) with \(\dim(\vect{S}')<\dim(\vect{S})\). To see this, take \(\vect{S}_i=\vect{c}+\vect{Q}\) such that \(\vect{X} \cap \vect{S}_i\) is reducible. We have \(\dim(\vect{S}_i \setminus \vect{x}+\vect{Q})= \dim(\vect{c}+\vect{Q} \setminus \vect{x}+\vect{Q}) <\dim(\vect{Q})\) by Lemma \ref{LemmaRemovingInnerConeReducesDimension}, wherefore the recursion uses a lower dimensional set, and termination follows from bounded recursion depth.

Correctness: Follows from correctness of Proposition \ref{PropositionFullLinearAlmostHybridlinearPartition}.

The partition is computable for VAS: We have to be able to find \(\vect{x}\) with \(\vect{x}+\vect{Q} \subseteq \vect{X}\) given the promise that such an \(\vect{x}\) exists. This is possible since containment of semilinear sets in reachability sets is decidable by \cite{Leroux13} using flatability.
\end{proof}

\section{Corollaries of Theorem \ref{TheoremFinalPartition}} \label{SectionFinalPartitionAndCorollaries}


\subsection{VAS semilinearity is decidable}
\label{subsec:SemilinearityIsDecidbale}

We reprove that the semilinearity problem for VAS is decidable. We start with a lemma, whose full proof is in the appendix.

\begin{restatable}{lemma}{LemmaSemilinearThenReducible}
Let \(\vect{X}\) be a semilinear Petri set with hybridization \(\vect{c}+\vect{Q}\). Then \(\vect{X}\) is reducible.\label{LemmaSemilinearThenReducible}
\end{restatable}%

\begin{proof}[Proof idea] 
The hybridization describes all ``limit directions'', with the problematic ones being for example ``north'' in case of the parabola \(\{(x,y) \mid y \leq x^2\}\), which is a limit but not actually a direction. If \(\vect{X}\) is semilinear though, then the steepness can only increase finitely often, namely when changing to a different linear component, and all limit directions are actually also directions. Using this for generators of \(\Fill(\vect{P})\) we find \(\vect{x}+\Fill(\vect{P}) \subseteq \vect{X}\).
\end{proof}

\begin{corollary}
\cite{Hauschildt90} The following problem is decidable.

Input: Reachability set \(\vect{R}\) of VAS, semilinear \(\vect{S}\).

Output: Is \(\vect{R} \cap \vect{S}\) semilinear?
\end{corollary}

\begin{proof}
As also mentioned in the introduction, the algorithm computes the partition of Theorem \ref{TheoremFinalPartition} and checks whether the third case does not occur.

Correctness: If \(\vect{R} \cap \vect{S}\) is semilinear, then in particular \(\vect{R} \cap \vect{S}_i\) is semilinear for every part \(\vect{S}_i\) of the partition. By Lemma \ref{LemmaSemilinearThenReducible}, \(\vect{R} \cap \vect{S}_i\) cannot be irreducible, and so either $\vect{R} \cap \vect{S}_i= \emptyset$ or $\vect{S}_i \subseteq \vect{R}$ for all \(i\). 

On the other hand, if only the cases \(\vect{R} \cap \vect{S}_i=\emptyset\) and \(\vect{S}_i \subseteq \vect{R}\) occur, then the $\vect{S}_i$ such that $\vect{S}_i \subseteq \vect{R}$ form a semilinear representation.
\end{proof}

\subsection{On the Complement of a VAS Reachability Set} \label{SubsectionInfiniteLineCorollary}

We show that if the complement of a VAS reachability set is infinite, then it contains an infinite linear set. The main part of the argument was already depicted in the middle of Figure \ref{FigureIntuitionAlmostHybridlinear}: If \(\vect{X}\) contains enough of the boundary, then it is reducible. 

We hence need to formalize the notion of boundary and interior also for full linear sets. If \(\vect{L}=\vect{b}+\vect{Q}\) is a full linear set, then \(\interior(\vect{L}):=\vect{b}+(\vect{Q} \cap \interior(\Q_{\geq 0}\vect{Q}))\) is the interior of \(\vect{L}\) and \(\partial(\vect{L}):=\vect{b}+(\vect{Q} \cap \partial(\Q_{\geq 0}\vect{Q}))\) is the boundary of \(\vect{L}\), both are inherited from the cone. These sets are both semilinear, as can be seen by using the definition expressible via \(\varphi \in \FO(\N, +, \geq)\), i.e. Presburger Arithmetic. Remember that we consider definable cones, i.e. cones expressible in \(\FO(\Q,+,\geq)\). In the appendix, we prove the following proposition, formalizing the first part of the proof.

\begin{restatable}{proposition}{PropositionBoundaryImpliesCone}
Let \(\vect{X}\) be a set with hybridization \(\vect{c}+\Fill(\vect{P})\). Assume that \(|\partial(\vect{c}+\Fill(\vect{P})) \setminus \vect{X}|<\infty\). Then \(\vect{X}\) is reducible. \label{PropositionBoundaryImpliesCone}
\end{restatable}

The proof of Proposition \ref{PropositionBoundaryImpliesCone} is illustrated in the above figure. The main difficulty is defining a ``wide enough'' cone \(\vect{C}\), then Proposition \ref{PropositionAlmostPeriodicityPeriodicSet} applied to \(\vect{C} \cap (\vect{P}-\vect{P})\) does the rest.

\begin{figure}[t]
\begin{minipage}{4.4cm}
\begin{tikzpicture}
\begin{axis}[
    xlabel={ },
    ylabel={ },
    xmin=0, xmax=8,
    ymin=0, ymax=8,
    xtick={0,2,4,6,8},
    ytick={0,2,4,6,8},
    ymajorgrids=true,
    xmajorgrids=true,
    thick,
    smooth,
    no markers,
]

\addplot[
    fill=blue,
    fill opacity=0.5,
    only marks,
    ]
    coordinates {
    (2,0)(4,0)(6,0)(8,0)(0,2)(0,4)(0,6)(0,8)
    };
    
\addplot[
    fill=red,
    fill opacity=0.5,
    only marks,
    ]
    coordinates {
    (0,0)
    };

\addplot[
    color=black,
    thick,
    ->,
    no marks,
    ]
    coordinates {
    (0.1,2.1)(1,3)
    };

\addplot+[
    name path=AB,
    color=black,
    domain=1:8,
    no marks,
]
{2*x+1};

\addplot+[
    name path=BB,
    color=black,
    domain=1:8,
    no marks,
]
{x/2 + 2.5};

\addplot[
    color=blue!40,
]
fill between[of=AB and BB];

\addplot[
    color=black,
    thick,
    ->,
    no marks,
    ]
    coordinates {
    (0.1,4.1)(1,5)
    };

\addplot+[
    name path=AC,
    color=black,
    domain=1:8,
    no marks,
]
{2*x+3};

\addplot+[
    name path=BC,
    color=black,
    domain=1:8,
    no marks,
]
{x/2 + 4.5};

\addplot[
    color=blue!40,
]
fill between[of=AC and BC];

\addplot[
    color=black,
    thick,
    ->,
    no marks,
    ]
    coordinates {
    (0.1,6.1)(1,7)
    };

\addplot+[
    name path=AD,
    color=black,
    domain=1:8,
    no marks,
]
{2*x+5};

\addplot+[
    name path=BD,
    color=black,
    domain=1:8,
    no marks,
]
{x/2 + 6.5};

\addplot[
    color=blue!40,
]
fill between[of=AD and BD];

\addplot[
    color=black,
    thick,
    ->,
    no marks,
    ]
    coordinates {
    (2.1,0.1)(3,1)
    };

\addplot+[
    name path=AE,
    color=black,
    domain=3:8,
    no marks,
]
{2*x-5};

\addplot+[
    name path=BE,
    color=black,
    domain=3:8,
    no marks,
]
{x/2 -0.5};

\addplot[
    color=blue!40,
]
fill between[of=AE and BE];

\addplot[
    color=black,
    thick,
    ->,
    no marks,
    ]
    coordinates {
    (4.1,0.1)(5,1)
    };

\addplot+[
    name path=AF,
    color=black,
    domain=5:8,
    no marks,
]
{2*x-9};

\addplot+[
    name path=BF,
    color=black,
    domain=5:8,
    no marks,
]
{x/2 -1.5};

\addplot[
    color=blue!40,
]
fill between[of=AF and BF];

\addplot[
    color=black,
    thick,
    ->,
    no marks,
    ]
    coordinates {
    (6.1,0.1)(7,1)
    };

\addplot+[
    name path=AG,
    color=black,
    domain=7:8,
    no marks,
]
{2*x-13};

\addplot+[
    name path=BG,
    color=black,
    domain=7:8,
    no marks,
]
{x/2 -2.5};

\addplot[
    color=blue!40,
]
fill between[of=AG and BG];

\addplot[
    color=black,
    thick,
    ->,
    no marks,
    ]
    coordinates {
    (0.1,0.1)(1,1)
    };

\addplot+[
    name path=AA,
    color=black,
    domain=1:8,
    no marks,
]
{2*x-1};

\addplot+[
    name path=BA,
    color=black,
    domain=1:8,
    no marks,
]
{x/2 + 0.5};

\addplot[
    color=red!40,
]
fill between[of=AA and BA];
\end{axis}
\end{tikzpicture}
\end{minipage}%
\begin{minipage}{9cm}
This refers to the figure on the left. 

Let \(\vect{C}\) be the cone generated by \((2,1)\) and \((1,2)\) and assume that \(\vect{X}+[(1,1)+\vect{C}] \subseteq \vect{X}\) holds. Then \((0,0) \in \vect{X}\) implies that the whole red shifted cone is in \(\vect{X}\). Importantly, we obtain a similar shifted cone for \emph{every} point \(\vect{x}' \in \vect{X}\). Hence if \(\partial \N^2 \subseteq \vect{X}\), then almost all of \(\N^2\) is contained in \(\vect{X}\). \label{FigureFinalArgument}
\end{minipage}
\end{figure}

\begin{corollary}
Let \(\vect{X}\) be a Petri set. Let \(\vect{S}\) be a semilinear set such that \(\vect{S} \setminus \vect{X}\) is infinite. Then $\vect{S} \setminus \vect{X}$ contains an infinite linear set. \label{CorollaryNonreachableLine}
\end{corollary}

\begin{proof}
Proof by induction on \(\dim(\vect{S})\). If \(\dim(\vect{S})=0\), the property holds vacuously. Else consider the partition of Theorem \ref{TheoremFinalPartition}. Since \(\vect{S} \setminus \vect{X}\) is infinite, some \(\vect{S}_i \setminus \vect{X}\) is infinite. Fix such an \(i\). Because of Theorem 1.1, \(\vect{S}_i\subseteq \vect{X}\) or \(\vect{X}\cap \vect{S}_i = \emptyset\) or \(\vect{X}\cap \vect{S}_i\) is irreducible. In fact, only the third possibility is interesting. If \(\vect{S}_i\subseteq \vect{X}\), then \(\vect{S}_i \setminus \vect{X}\) can not be infinite. If \(\vect{S}_i \cap \vect{X}= \emptyset\) then \(\vect{S}_i=\vect{S}_i\setminus \vect{X}\), hence it contains a line. Let us consider the case when \(\vect{S}_i\cap \vect{X}$ is irreducible.
Assume for contradiction that \(\vect{S}_i \setminus \vect{X}\) does not contain an infinite linear set. Then in particular \(\partial(\vect{S}_i) \setminus \vect{X}\) does not. We have \(\dim(\partial(\vect{S_i}))<\dim(\vect{S_i})\), since the boundary is contained in the finite union of the facets. Hence \(|\partial(\vect{S}_i) \setminus \vect{X}|<\infty\) by induction. By Proposition \ref{PropositionBoundaryImpliesCone}, \(\vect{X} \cap \vect{S}_i\) is reducible. Contradiction.
\end{proof}

In the appendix we even prove another corollary of the partition. The proof is based on the existence of a partition as in Theorem \ref{TheoremFinalPartition}, which has the properties for two Petri sets \(\vect{X}_1\) and \(\vect{X}_2\) at once.

\begin{restatable}{corollary}{CorollarySeparatingPairOfPetriSets}
Let \(\vect{X}_1\) and \(\vect{X}_2\) be Petri sets with \(\vect{X}_1 \cap \vect{X}_2=\emptyset\). Then there exists a semilinear set \(\vect{S}'\) such that \(\vect{X}_1 \subseteq \vect{S}'\) and \(\vect{X}_2 \cap \vect{S}'=\emptyset\). \label{CorollarySeparatingPetriSets}
\end{restatable}

\begin{restatable}{corollary}{CorollarySeparatingPetriSet}
Let \(\mathcal{V}\) be a VAS, and \(\vect{X}\) a Petri set such that \(\Reach(\mathcal{V}) \cap \vect{X}=\emptyset\). Then there exists a semilinear inductive invariant \(\vect{S}'\) of \(\mathcal{V}\) such that \(\Reach(\mathcal{V}) \subseteq \vect{S}'\) and \(\vect{X} \cap \vect{S}'=\emptyset\).
\end{restatable}

\section{Conclusion}


We have introduced hybridizations, and used them to prove a powerful decomposition theorem for Petri sets. For VAS reachabillity sets the decomposition can be effectively computed. We have derived several geometric and computational results. We think that our decomposition can help to study the computational power of VAS. For example, it leads to this corollary:

\begin{corollary}
Let \(f \colon \N \to \N\) be a function whose graph does not contain an infinite line. Then either \(\{(x,y) \mid y < f(x)\}\) or \(\{(x,y) \mid y > f(x)\}\) is not a Petri set.
\end{corollary}

\begin{proof}
Assume for contradiction that both are Petri sets. Then, since finite unions of Petri sets are again Petri sets, \(\{(x,y) \mid y \neq f(x)\}\) is a Petri set. Its complement is the graph of \(f\), which by assumption does not contain an infinite line. Contradiction to Corollary \ref{CorollaryNonreachableLine}.
\end{proof}

We plan to study other possible applications of our result, derived from the fact that the reachability relation of a VAS is also a Petri set.

\bibliography{geometry_of_VAS}

\appendix


\section{Proofs of Section \ref{Sec:Preliminaries}}

\LemmaEquivalentDefinitionOfFull*

\begin{proof}
``\(\Leftarrow\)'': \(\vect{P}=\Fill(\vect{P})\) is by definition an intersection of a f.g. cone and a lattice.

``\(\Rightarrow\)'': Write \(\vect{P}=\vect{C} \cap \vect{L}\). We first claim that \(\Q_{\geq 0}\vect{P}=\vect{C} \cap \Q_{\geq 0}\vect{L}\). As intersection of f.g. cones, \(\Q_{\geq 0}\vect{P}\) is then finitely generated (and hence closed), use for example Lemma \ref{LemmaFinitelyGeneratedCones}.

\begin{proof}[Proof of Claim] ``\(\subseteq\)'' is clear. Hence let \(\vect{x} \in \vect{C} \cap \Q_{\geq 0}\vect{L}\). Then there exists \(\lambda \in \N\) such that \(\lambda \vect{x}\in \vect{L}\). Then \(\lambda \vect{x} \in \vect{C} \cap \vect{L}=\vect{P}\) as claimed.
\end{proof}

It is left to prove \(\Q_{\geq 0}\vect{P} \cap (\vect{P}-\vect{P})=\vect{P}\). ``\(\supseteq\)'' is clear, hence let \(\vect{x}\in \Q_{\geq 0}\vect{P} \cap (\vect{P}-\vect{P})\). It is enough to prove \(\vect{x} \in \vect{C}\) and \(\vect{x}\in \vect{L}\). To prove those inclusions, observe that \(\vect{x}\in \Q_{\geq 0}\vect{P} \subseteq \Q_{\geq 0} \vect{C} \subseteq \vect{C}\) and \(\vect{x}\in \vect{P}-\vect{P} \subseteq \vect{L}-\vect{L} \subseteq \vect{L}\).
\end{proof}

\LemmaFullThenLessPeriods*

\begin{proof}
Since \(\vect{P}\) is full, by Lemma \ref{LemmaEquivalentDefinitionOfFull} it is sufficient to prove \(\vect{Q} \subseteq \vect{P}-\vect{P}\) and \(\vect{Q} \subseteq \overline{\mathbb{Q}_{\geq 0}\vect{P}}\). 

To prove \(\vect{Q} \subseteq \vect{P}-\vect{P}\), observe that \(\vect{Q}=(\vect{c}+\vect{Q}) - \vect{c}  \subseteq (\vect{b}+\vect{P})-(\vect{b}+\vect{P})=\vect{P}-\vect{P}\). 

To prove \(\vect{Q} \subseteq \overline{\mathbb{Q}_{\geq 0}\vect{P}}\), write \(\overline{\Q_{\geq 0}\vect{P}}=\{\vect{x} \in \VectorSpace(\vect{P}) \mid A \vect{x} \geq 0\}\) for a matrix \(A\), as in Lemma \ref{LemmaFinitelyGeneratedCones}. Let \(A_k\) be the \(k\)-th row of \(A\). It suffices to show \(A_k \vect{x} \geq 0\) for all \(\vect{x} \in \vect{Q}\). If we had \(A_k \vect{x} <0\), then \(A_k (\vect{c}+\lambda \vect{x}) < A_k \vect{b}\) for large enough \(\lambda\), contradicting \(\vect{c}+\vect{Q} \subseteq \vect{b}+\vect{P}\).
\end{proof}

\LemmaRemovingInnerConeReducesDimension*

\begin{proof}
We only prove the case where \(\vect{P}\) is full, i.e.\ \(\Fill(\vect{P})=\vect{P}\), we do not need the full lemma. First of all, since semilinear sets are closed under all boolean operations, \(\vect{S}\) is semilinear. Secondly, we have \(\vect{P} \cap [\vect{x} + \mathbb{Q}_{\geq 0}\vect{P}] \subseteq \vect{x} + \Fill(\vect{P})\), i.e. every point of \(\vect{P}\) which is in the inner cone is removed for \(\vect{S}\). To see this, let \(\vect{y}\in \vect{x} + \mathbb{Q}_{\geq 0}\vect{P}\). Write \(\vect{y}=\vect{x} + \vect{v}\). We then have \(\vect{v}\in \mathbb{Q}_{\geq 0} \vect{P}\) by definition, and \(\vect{v}\in \vect{P}-\vect{P}\) since it is the difference of \(\vect{y}\in \vect{P}\) and \(\vect{x} \in \vect{P}\). Therefore \(\vect{v}\in \Fill(\vect{P})\) and \(\vect{y}\in \vect{x} + \Fill(\vect{P})\).

Since \(\mathbb{Q}_{\geq 0}\vect{P}\) is finitely generated, by Lemma \ref{LemmaFinitelyGeneratedCones}, there exists an integer matrix \(A\) such that \(\Q_{\geq 0} \vect{P}=\{\vect{y} \in \VectorSpace(\vect{P}) \mid A \vect{y} \geq 0\}\). Let \(A_i\) be the \(i\)-th row of \(A\). Since by the above, \(\vect{S} \cap \vect{x} + \mathbb{Q}_{\geq 0}\vect{P}=\emptyset\), every point \(y\in S\) fulfills \(0 \leq A_i \vect{y} \leq A_i \cdot \vect{x}\) for at least one \(i\). Since \(A_i \cdot \vect{y}\) can only take integer values, we define \(\vect{V}_{i,j}:=\{\vect{y} \in \VectorSpace(\vect{P}) \mid A_i \cdot \vect{y}=j\}\) for all \(j \in \{0, \dots, A_i \cdot \vect{x}\}\) and obtain that \(\vect{S} \subseteq \bigcup_{i=1}^r \bigcup_{j=0}^{A_i \cdot \vect{x}} \vect{V}_{i,j}\). Since every \(\vect{V}_{i,j}\) has co-dimension \(1\) in \(\VectorSpace(\vect{P})\), we obtain \(\dim(\vect{S})\leq \dim(\VectorSpace(\vect{P}))-1=\dim(\vect{P})-1\) as claimed.
\end{proof}

\BasicDimensionProperties*

\begin{proof}
1) We prove \(\dim(\vect{X}) \geq \dim(\vect{b}+\vect{X})\), the other direction follows by choosing \(-\vect{b}\). Let \(\vect{b}_i\) and \(\vect{V}_i\) such that \(\dim(\vect{V}_i)\leq \dim(\vect{X})\) and \(\vect{X} \subseteq \bigcup_{i=1}^r \vect{b}_i+\vect{V}_i\). Then \(\vect{b}+\vect{X} \subseteq \bigcup_{i=1}^r (\vect{b}+\vect{b}_i)+\vect{V}_i\), and hence \(\dim(\vect{b}+\vect{X})\leq \max_i \dim(\vect{V}_i)=\dim(\vect{X})\).

2) ``\(\geq\)'': Covering \(\vect{X} \cup \vect{X}'\) covers both \(\vect{X}\) and \(\vect{X}'\).

``\(\leq\)'': Let \(\vect{b}_i\) and \(\vect{V}_i\) such that \(\dim(\vect{V}_i) \leq \dim(\vect{X})\) and \(\vect{X} \subseteq \bigcup_{i=1}^r \vect{b}_i+\vect{V}_i\), and \(\vect{c}_i\) and \(\vect{W}_i\) such that \(\dim(\vect{W}_i) \leq \dim(\vect{X}')\) and \(\vect{X}' \subseteq \bigcup_{i=1}^s \vect{c}_i+\vect{W}_i\). Then \[\vect{X} \cup \vect{X}' \subseteq \bigcup_{i=1}^r \vect{b}_i+\vect{V}_i \cup \bigcup_{i=1}^s \vect{c}_i+\vect{W}_i.\]

3) Use 2), and observe that \(\vect{X} \cup \vect{X}'=\vect{X}'\).
\end{proof}

\section{Proofs of Section \ref{Sec:DirectionsOfAPeriodicSet}}

\PropositionIntersectionPeriodicSets*

We split the proof of Proposition \ref{PropositionIntersectionPeriodicSets} in a lemma for every part, with \(\Fill(\vect{P}\cap \vect{P}')\) having one lemma for the lattice and one lemma for the closed cone.

\subsubsection*{Proof of Proposition \ref{PropositionIntersectionPeriodicSets}(1)}

\begin{lemma}
Let \(\vect{P},\vect{P}'\) be smooth periodic sets such that
\[\dim(\Fill(\vect{P}) \cap \Fill(\vect{P}'))=\dim(\Fill(\vect{P}))=\dim(\Fill(\vect{P}')).\]

Then \(\dim(\vect{P} \cap \vect{P}')=\dim(\vect{P})=\dim(\vect{P}')\).\label{LemmaImproveDimensionCondition}
\end{lemma}

\begin{proof}
We only argue \(\dim(\vect{P} \cap \vect{P}')=\dim(\vect{P})\), the other equality follows by symmetry. ``\(\leq\)'' is immediate.

By Lemma \ref{LemmaFromJerome} we have \(\dim(\vect{P})=\dim(\VectorSpace(\vect{P}))\), in particular also \(\dim(\vect{P})=\dim(\overline{\mathbb{Q}_{\geq 0}\vect{P}})=\dim(\Fill(\vect{P}))\). Hence \(\dim(\vect{P})=\dim(\Fill(\vect{P}))=\dim(\Fill(\vect{P}) \cap \Fill(\vect{P}'))\leq \dim(\overline{\mathbb{Q}_{\geq 0}\vect{P}} \cap \overline{\mathbb{Q}_{\geq 0}\vect{P}'})\). The goal is to arrive at \(\dim(\mathbb{Q}_{\geq 0}\vect{P} \cap \mathbb{Q}_{\geq 0}\vect{P}')\). Then, since \(\Q_{\geq 0}(\vect{P}\cap \vect{P}')=\Q_{\geq 0}\vect{P} \cap \Q_{\geq 0}\vect{P}'\) always holds \cite[Lemma 4.5]{Leroux11}, we would obtain \(\dim(\vect{P}) \leq \dim(\Q_{\geq 0}(\vect{P} \cap \vect{P}'))=\dim(\vect{P}\cap \vect{P}')\), and be done. 

One main tool is Lemma \ref{LemmaDirectionsContainInterior}. For definable cones \(\vect{C}\), we have \(\dim(\partial(\vect{C}))<\dim(\vect{C})\), since the boundary consists of finitely many facets, which are lower dimensional cones. In addition \(\vect{C}=\partial(\vect{C})\cup \interior(\vect{C})\), hence \(\dim(\vect{C})=\dim(\interior(\vect{C}))\) by Lemma \ref{BasicDimensionProperties}. We will use these two facts for \(\vect{C}:=\overline{\Q_{\geq 0}\vect{P}}\) and \(\vect{C}':=\overline{\Q_{\geq 0}\vect{P}'}\). 

We have \(\interior(\vect{C}) \cap \interior(\vect{C}')=(\vect{C} \setminus \partial(\vect{C})) \cap (\vect{C}' \setminus \partial(\vect{C}'))=(\vect{C} \cap \vect{C}') \setminus (\partial(\vect{C}) \cup \partial(\vect{C}'))\), which has dimension \(\dim(\vect{C} \cap \vect{C}')\), since we only remove a lower dimensional set. In total we can finish the chain from above:
\begin{align*}
\dim(\vect{P})&\leq \dim(\vect{C} \cap \vect{C}')= \dim(\interior(\vect{C}) \cap \interior(\vect{C}')) \leq \dim(\Q_{\geq 0}\vect{P} \cap \Q_{\geq 0}\vect{P}') = \dim(\vect{P} \cap \vect{P}'),
\end{align*}
where we had already argued the last step earlier.
\end{proof}

\subsubsection*{Proof of Proposition \ref{PropositionIntersectionPeriodicSets}(2)}

We first prove an auxiliary lemma.

\begin{lemma}
Let \(\vect{P} \subseteq \vect{P}'\) be periodic sets with \(\dim(\vect{P})=\dim(\vect{P}')\). Then for all \(\vect{p}_1 \in \vect{P}'\) there exists a \(\vect{p}_2 \in \vect{P}'\) such that \(\vect{p}_1+\vect{p}_2 \in \vect{P}\). \label{LemmaSameDimensionPeriodicSets}
\end{lemma}

\begin{proof}
Let \(\vect{p}_1 \in \vect{P}'\). Since \(\dim(\vect{P})=\dim(\vect{P}')\) and \(\vect{P} \subseteq \vect{P}'\), by Lemma \ref{LemmaFromJerome} they generate the same vector space \(\vect{V}\). Hence there exists \(\lambda \in \N_{>0}\) such that \(\lambda \vect{p}_1 \in \vect{P}-\vect{P}\), namely \(\lambda\) is the common denominator for some rational linear combination of elements from \(\vect{P}\). By writing \(\lambda \vect{p}_1 =\vect{p}-\vect{p}'\) with \(\vect{p},\vect{p}' \in \vect{P}\) and choosing \(\vect{p}_2:=(\lambda-1) \cdot \vect{p}_1 +\vect{p}' \in \vect{P}'\), we obtain \(\vect{p}_1+\vect{p}_2=\vect{p} \in \vect{P}\).
\end{proof}
\begin{restatable}{lemma}{LemmaDirectionsIntersectCorrectly}
Let \(\vect{P}, \vect{P}'\) be periodic sets such that 

\(\dim(\vect{P} \cap \vect{P}')=\dim(\vect{P})=\dim(\vect{P}')\). 

Then \(\dir(\vect{P} \cap \vect{P}')=\dir(\vect{P}) \cap \dir(\vect{P}')\).

We even have: If \(\vect{x}+\N \cdot \vect{d} \subseteq \vect{P}\) and \(\vect{x}'+\N \cdot \vect{d} \subseteq \vect{P}'\) for some \(\vect{x},\vect{x}'\), then \(\vect{x}''+\N \cdot \vect{d} \subseteq \vect{P} \cap \vect{P}'\) for some \(\vect{x}''\). \label{LemmaPhysicalDirectionsIntersectCorrectly}
\end{restatable}

\begin{proof}
``\(\subseteq\)'' is clear. Therefore let \(\vect{d} \in \dir(\vect{P}) \cap \dir(\vect{P}')\). Then there exist \(m, m' \in \N_{>0}\) and \(\vect{x},\vect{x}'\) such that \(\vect{x}+\N \cdot m \vect{d} \subseteq \vect{P}\) and \(\vect{x}' + \N \cdot m' \vect{d} \subseteq \vect{P}'\). Replacing \(\vect{d}\) by \(m m' \vect{d}\), we obtain \(\vect{x}+ \N \cdot \vect{d} \subseteq \vect{P}\) and \(\vect{x}' + \N \cdot \vect{d} \subseteq \vect{P}'\).

By Lemma \ref{LemmaSameDimensionPeriodicSets}, there exists \(\vect{p} \in \vect{P}\) such that \(\vect{x}+\vect{p} \in \vect{P} \cap \vect{P}'\), and again by the same Lemma \(\vect{p}' \in \vect{P}'\) such that \(\vect{x}'+\vect{p}' \in \vect{P} \cap \vect{P}'\). We choose \(\vect{x}'':=(\vect{x}+\vect{p})+(\vect{x}'+\vect{p}') \in \vect{P} \cap \vect{P}'\), and obtain that \(\vect{x}''+\N \cdot \vect{d} \subseteq \vect{P}\), because \[\vect{x}''+\N \cdot \vect{d}=((\vect{x}+\N \cdot \vect{d})+\vect{p})+(\vect{x}'+\vect{p}') \subseteq \vect{P}+(\vect{P} \cap \vect{P}') \subseteq \vect{P}.\] Symmetrically we also obtain \(\vect{x}''+\N \cdot \vect{d} \subseteq \vect{P}'\).
\end{proof}

\subsubsection*{Proof of Proposition \ref{PropositionIntersectionPeriodicSets}(3)}

Recall that \(\Fill(\vect{P}):=(\vect{P}-\vect{P}) \cap \overline{\mathbb{Q}_{\geq 0}\vect{P}}\). So it suffices to prove that the 
closed cone and the lattice of an intersection are equal to the intersection of the closed cones and the lattices respectively.

%
%
%

\begin{lemma}
Let \(\vect{P}, \vect{P}'\) be smooth periodic sets with 

\(\dim(\vect{P} \cap \vect{P}')=\dim(\vect{P})=\dim(\vect{P}')\). 

Then \(\overline{\mathbb{Q}_{\geq 0}(\vect{P} \cap \vect{P}')}=\overline{\mathbb{Q}_{\geq 0}\vect{P}} \cap \overline{\mathbb{Q}_{\geq 0}\vect{P}'}\).
\label{LemmaConesIntersectCorrectly}
\end{lemma}

\begin{proof}
Define \(\vect{C}:=\dir(\vect{P})\) and \(\vect{C}':=\dir(\vect{P}')\). Since \(\vect{P}\) and \(\vect{P}'\) are asymptotically definable, we have \(\vect{C} \setminus \{\vect{0}\}=\{\vect{x} \in \Vectorspace(\vect{P}) \mid A_1 \vect{x} >\vect{0}, A_2 \vect{x} \geq \vect{0}\}\) and similarly \(\vect{C}' \setminus \{\vect{0}\}=\{\vect{x} \in \Vectorspace(\vect{P}) \mid A_1' \vect{x} >\vect{0}, A_2' \vect{x} \geq \vect{0}\}\). By Lemma \ref{LemmaFromJerome}, we have \(\dim(\Vectorspace(\vect{P}))=\dim(\vect{C})=\dim(\vect{C}')=\dim(\vect{C} \cap \vect{C}')\). For equations defining a set of full dimension, taking the closure of the cone is equivalent to changing all equations to \(\geq 0\). We hence have \(\overline{\vect{C}}=\{\vect{x} \in \Vectorspace(\vect{P}) \mid A_1 \vect{x} \geq \vect{0}, A_2 \vect{x} \geq \vect{0}\}\), as well as \(\overline{\vect{C}'}=\{\vect{x} \in \Vectorspace(\vect{P}) \mid A_1' \vect{x} \geq \vect{0}, A_2' \vect{x} \geq \vect{0}\}\) and \(\overline{\vect{C} \cap \vect{C}'}=\{\vect{x} \in \Vectorspace(\vect{P}) \mid A_1 \vect{x} \geq \vect{0}, A_2 \vect{x} \geq \vect{0}, A_1' \vect{x} \geq \vect{0}, A_2' \vect{x} \geq \vect{0}\}\).
\end{proof}

\begin{lemma}
Let \(\vect{P}, \vect{P}'\) be periodic sets with 

\(\dim(\vect{P} \cap \vect{P}')=\dim(\vect{P})=\dim(\vect{P}')\). 

Then \((\vect{P} \cap \vect{P}')-(\vect{P} \cap \vect{P}')=(\vect{P}-\vect{P}) \cap (\vect{P}'-\vect{P}')\).\label{LemmaLatticesIntersectCorrectly}
\end{lemma}

\begin{proof}
``\(\subseteq\)'' is clear, hence let \(\vect{d} \in (\vect{P}-\vect{P}) \cap (\vect{P}'-\vect{P}')\). If lattices coincide on a cone of full dimension like \(\dir(\vect{P} \cap \vect{P}')=\dir(\vect{P}) \cap \dir(\vect{P}')\), then they coincide everywhere. We can therefore assume \(\vect{d} \in \dir(\vect{P} \cap \vect{P}')\).

By Proposition \ref{PropositionAlmostPeriodicityPeriodicSet}, we have \(\vect{x}+\N \cdot \vect{d} \subseteq \vect{P}\) and \(\vect{x}' + \N \cdot \vect{d} \subseteq \vect{P}'\) for some \(\vect{x},\vect{x}'\). By the extra remark in Lemma \ref{LemmaPhysicalDirectionsIntersectCorrectly}, there exists \(\vect{x}''\) such that \(\vect{x}''+\N \cdot \vect{d} \subseteq \vect{P} \cap \vect{P}'\). Hence
\(\vect{d}=(\vect{x}''+\vect{d})-(\vect{x}'') \in (\vect{P} \cap \vect{P}')-(\vect{P} \cap \vect{P}').\)
\end{proof}

\begin{corollary}
Let \(\vect{P}, \vect{P}'\) be smooth periodic sets with 

\(\dim(\vect{P} \cap \vect{P}')=\dim(\vect{P})=\dim(\vect{P}')\). 

Then \(\Fill(\vect{P} \cap \vect{P}')=\Fill(\vect{P}) \cap \Fill(\vect{P}')\).
\end{corollary}
\begin{proof}
Follows from Lemmas \ref{LemmaConesIntersectCorrectly} and \ref{LemmaLatticesIntersectCorrectly}.
\end{proof}

\subsubsection*{Proof of Proposition \ref{PropositionIntersectionPeriodicSets}(4)}

\begin{lemma}
Let \(\vect{P}, \vect{P}'\) be smooth periodic sets with

\(\dim(\vect{P} \cap \vect{P}')=\dim(\vect{P})=\dim(\vect{P}')\). 

Then \(\vect{P} \cap \vect{P}'\) is smooth. \label{LemmaIntersectionAgainSmooth}
\end{lemma}

\begin{proof}
\(\vect{P} \cap \vect{P}'\) is asymptotically definable by Lemma \ref{LemmaPhysicalDirectionsIntersectCorrectly}. Let \((\vect{p}_m + \vect{p}_m')_m\) with \(\vect{p}_m \in \vect{P}\) and \(\vect{p}_m' \in \vect{P}'\) be a sequence. Since \(\vect{P}\) is well-directed, there exists an infinite set of indices \(N_1 \subseteq \N\) such that \(\vect{p}_m-\vect{p}_k \in \dir(\vect{P})\) for all \(m>k\) in \(N_1\). Since \(\vect{P}'\) is well-directed, there exists an infinite set \(N_2 \subseteq N_1\) such that moreover \(\vect{p}_m'-\vect{p}_k' \in \dir(\vect{P}')\) for all \(m>k\) in \(N_2\). Hence for all \(m>k\) in \(N_2\), we have \((\vect{p}_m+\vect{p}_m')-(\vect{p}_k+\vect{p}_k') \in \dir(\vect{P})+\dir(\vect{P}') \subseteq \dir(\vect{P}+\vect{P}')\), since \(\dir(\vect{P}+\vect{P}')\) is closed under addition.
\end{proof}

\section{Proofs of Section \ref{Sec:NewPeriodicityProperty}}

Before we can prove Theorem \ref{TheoremEquivalentAlmostHybridlinearCondition}, we need preliminary lemmas.

\begin{restatable}{lemma}{LemmaSumPreservesSmooth}
Let \(\vect{P},\vect{P}'\) be smooth periodic sets. Then \(\vect{P}+\vect{P}'\) is smooth with \(\dir(\vect{P}+\vect{P}')=\dir(\vect{P})+\dir(\vect{P}')\). \label{LemmaSumPreservesSmooth}
\end{restatable}

\begin{proof}
Let \(\vect{p}_1+\vect{p}_1', \vect{p}_2+\vect{p}_2' \in \vect{P}+\vect{P}'\). Then \((\vect{p}_1+\vect{p}_1')+(\vect{p}_2+\vect{p}_2')=(\vect{p}_1+\vect{p}_2)+(\vect{p}_1'+\vect{p}_2')\in \vect{P}+\vect{P}'\), i.e. \(\vect{P}+\vect{P}'\) is a periodic set. Next we show \(\dir(\vect{P})+\dir(\vect{P}') = \dir(\vect{P}+\vect{P}')\), where \(\subseteq\) is clear. Hence let \(\vect{d}\in \dir(\vect{P}+\vect{P}')\). Then there exists \(\vect{x}\in \vect{P}+\vect{P}'\) and \(\lambda \in \N_{>0}\) such that \(\vect{x}+\N \cdot \lambda \vect{d} \subseteq \vect{P}+\vect{P}'\). We write \(\vect{x}+m \lambda \vect{d}=\vect{p}_m+\vect{p}_m'\) with \(\vect{p}_m \in \vect{P}\) and \(\vect{p}_m' \in \vect{P}'\) to obtain the sequences \((\vect{p}_m)_m\) and \((\vect{p}_m')_m\). Since \(\vect{P}\) is well-directed, there exists an infinite set of indices \(N_1 \subseteq \N\) such that \(\vect{p}_m - \vect{p}_k \in \dir(\vect{P})\) for all \(m>k\) in \(N_1\). Now consider the sequence \((\vect{p}_m')_{m \in N_1}\). Since \(\vect{P}'\) is well-directed, there exists an infinite set of indices \(N_2 \subseteq N_1\) such that furthermore \(\vect{p}_m'-\vect{p}_k' \in \dir(\vect{P}')\) for all \(m>k\) in \(N_2\). Choose \(m>k \in N_2\). We obtain that 
\begin{align*}
(m-k) \lambda \vect{d} &=(\vect{p}_m+\vect{p}_m')-(\vect{p}_k+\vect{p}_k') \\
&=(\vect{p}_m-\vect{p}_k)+(\vect{p}_m'-\vect{p}_k') \in \dir(\vect{P})+\dir(\vect{P}').
\end{align*}

Hence also \(\vect{d}\in \dir(\vect{P})+\dir(\vect{P}')\). Proving that \(\vect{P}+\vect{P}'\) is well-directed similarly relies upon \(N_2 \subseteq N_1 \subseteq \N\).
\end{proof}

\begin{restatable}{lemma}{FullAfterFill}
Let \(\vect{P}\) be a smooth periodic set. Then \(\Fill(\Fill(\vect{P}))=\Fill(\vect{P})\). \label{LemmaFullAfterFill}
\end{restatable}

\begin{proof}
Since \(\vect{P}\) is smooth, \(\Fill(\vect{P})\) is full. Hence, by Lemma \ref{LemmaEquivalentDefinitionOfFull}, \(\vect{P}':=\Fill(\vect{P})\) fulfills \(\Fill(\vect{P}')=\vect{P}'\), as required.
%
\end{proof}

\TheoremEquivalentAlmostHybridlinearCondition*

\begin{proof}
``\(\Rightarrow\)'': Write \(\vect{X}=\bigcup_{i=1}^r \vect{b}_i + \vect{P}_i\) where the \(\vect{P}_i\) are smooth periodic sets with \(\Fill(\vect{P}_i)=\Fill(\vect{P}_j)\). Define \(\vect{P}:= \bigcap_{i=1}^r \vect{P}_i\), which is smooth by Proposition \ref{PropositionIntersectionPeriodicSets}, and we have \[ \Fill(\vect{P})=\bigcap_{i=1}^r \Fill(\vect{P}_i)=\bigcap_{i=1}^r \Fill(\vect{P}_1)=\Fill(\vect{P}_1).\] Since \(\vect{P}_i+\vect{P} \subseteq \vect{P}_i\) for all \(i\), we have \(\vect{X}+\vect{P} \subseteq \vect{X}\). We also have \(\vect{X} \subseteq \vect{H}:=\{\vect{b}_1, \dots, \vect{b}_r\}+\Fill(\vect{P})\).

``\(\Leftarrow\)'': Let \(\vect{P}\) smooth such that \(\vect{X}+\vect{P} \subseteq \vect{X}\) and the weak hybridization is \(\vect{H}=\{\vect{b}_1, \dots, \vect{b}_r\}+\Fill(\vect{P})\). 
Since \(\vect{X}\) is a Petri set, the sets \(\vect{X} \cap [\vect{b}_i+\Fill(\vect{P})]\) are almost semilinear, and so \(\vect{X} \cap [\vect{b}_i + \Fill(\vect{P})] =\bigcup_{j=1}^{r_i} \vect{b}_{i,j} + \vect{P}_{i,j}\) for smooth periodic sets \(\vect{P}_{i,j}\). Since \(\vect{X}+\vect{P} \subseteq \vect{X} \), we have \(\vect{X}=\bigcup_{i=1}^r \bigcup_{j=1}^{r_i} \vect{b}_{i,j} + (\vect{P}_{i,j} +\vect{P})\). We prove that this is an almost hybridlinear representation of $\vect{X}$ (see Definition \ref{def:almosthybridlinear}). By Lemma \ref{LemmaSumPreservesSmooth}, all sets $(\vect{P}_{i,j} +\vect{P})$ are smooth. 
We prove that all their fills are equal to $\Fill(\vect{P})$. 

It suffices to show \(\Fill(\vect{P}_{i,j}+\vect{P})\subseteq \Fill(\vect{P})\), the other inclusion is trivial. 

By Lemma \ref{LemmaFullThenLessPeriods}, we have \(\vect{P}_{i,j} \subseteq \Fill(\vect{P})\). By Lemma \ref{LemmaFullAfterFill}, we then also have $\Fill(\vect{P}_{i,j}) \subseteq \Fill(\vect{P})$. Since $\Fill(\vect{P})$ is periodic,
$\Fill(\vect{P}_{i,j}) + \Fill(\vect{P}) \subseteq \Fill(\vect{P})$ and so, again by Lemma \ref{LemmaFullAfterFill}, $\Fill(\Fill(\vect{P}_{i,j})+\Fill(\vect{P})) \subseteq \Fill(\vect{P})$.
Since $\vect{P}_{i,j} + \vect{P} \subseteq \Fill(\vect{P}_{i,j})+\Fill(\vect{P})$, we are done.
\end{proof}

\section{Proofs of Section \ref{Sec:PetriSetsAndDimension}}

\subsection{Proof of Theorem \ref{TheoremReducibilityIsDecidable}}

The starting point for this section is Theorem \ref{TheoremEquivalentAlmostHybridlinearCondition}, which allows us to consider only the case of almost hybridlinear sets. We prove an equivalent condition of reducibility for this case. As running example for this section we consider \(\vect{X}:=\vect{X}_1 \cup \vect{X}_2:=\vect{P}_1 \cup \vect{P}_2\) for \(\vect{P}_1=\{(x,y) \in \N^2 \mid y \leq x^2\}\) and \(\vect{P}_2=\{(x,y) \in \N^2 \mid y \geq \log_2(x+1)+3\} \cup \{(0,0)\}\) (see the middle of Figure \ref{FigureIntuitionAlmostHybridlinear}). Together, these two almost linear components have the hybridization \(\N^2\) and even though neither of the two components is reducible, the union is. 

We will provide an equivalent definition of reducibility in terms of a concept from \cite{Leroux13} called complete extraction. This definition will be more suited for an algorithmic check. 

\begin{definition}
\label{def:CompleteExtraction}
Let \(\mathcal{K}=\{\vect{K}_1, \dots, \vect{K}_r\}\) be a finite set of cones. A complete extraction of \(\mathcal{K}\) is a set of finitely generated cones \(\{\vect{C}_1, \dots, \vect{C}_r\}\) such that \(\vect{C}_i \subseteq \vect{K}_i\) for all \(i\) and \(\bigcup_{i=1}^r \vect{C}_i =\bigcup_{i=1}^r \vect{K}_i\). 
\end{definition}

Intuitively,  we try to replace the non-finitely generated \(\vect{K}_i\) by smaller cones \(\vect{C}_i\) which are finitely generated, but whose union is still the same. This is of course only possible if the cones \(\vect{K}_i\) have an overlap. The following lemma, proved in \cite{Leroux13}, in some sense formalizes this intuition and will help us prove that the existence of a complete extraction is decidable.

\begin{lemma}{\cite[App. E]{Leroux13}}
A finite set \(\mathcal{K}=\{\vect{K}_1, \dots, \vect{K}_r\}\) of cones has a complete extraction if and only if for all vectors \(\vect{v}_1, \dots, \vect{v}_s\) such that \(\mathbb{Q}_{>0}\vect{v}_1 + \dots + \mathbb{Q}_{>0}\vect{v}_s \subseteq \vect{K}_i\) for some \(i\), there exists a \(j \in \{1, \dots, r\}\) such that \(\vect{K}_j \cap (\mathbb{Q}_{>0}\vect{v}_1 + \dots + \mathbb{Q}_{>0}\vect{v}_k) \neq \emptyset\) for all \(1 \leq k \leq s\).\label{CharOfCompleteExtraction}
 \end{lemma}

\begin{example}
The set \(\vect{X}=\vect{X}_1 \cup \vect{X}_2\) fulfils \(\dir(\vect{P}_1)=\{(x,y) \mid x>0\}\) and \(\dir(\vect{P}_2)=\{(x,y) \mid y>0\}\). The set \(\{\dir(\vect{P}_1), \dir(\vect{P}_2)\}\) has a complete extraction, for example \(\vect{C}_1:=\{(x,y)\mid x \geq y\}\) and \(\vect{C}_2:=\{(x,y) \mid y \geq x\}\). If \(\dir(\vect{P}_2)\) were only \(\{(x,y) \mid x=0\}\), then \(\vect{K}_1 \cup \vect{K}_2=\mathbb{Q}_{\geq 0}^2\) would still hold, but there would be no complete extraction. Intuitively, \(\vect{C}_1\) would have to contain the ``open border'' \(x>0\) in that case. Remember that finitely generated cones have to be closed however.

The complete extraction will translate to 

\((10,10)+(\vect{C}_1 \cap \N^2) \subseteq \vect{X}_1\) and \((10,10)+(\vect{C}_2 \cap \N^2) \subseteq X_2\), 

i.e. \(\{(x,y) \in \N^2 \mid x\geq 10, y \geq 10\} \subseteq \vect{X}\), proving that \(\vect{X}\) is reducible.
\end{example}

The equivalent definition of reducibility is as follows.

\begin{restatable}{theorem}{TheoremEquivalentReducibilityCondition}
Let \(\vect{X}\) be almost hybridlinear with hybridization \(\vect{c}+\Fill(\vect{P})\). Write \(\vect{X}=\bigcup_{i=1}^r \vect{b}_i+\vect{P}_i\), where the sets $\vect{P}_1, \ldots, \vect{P}_r$ satisfy $\Fill(\vect{P})=\Fill(\vect{P}_1) = \cdots = \Fill(\vect{P}_r)$. Define \(\vect{K}_i:=\dir(\vect{P}_i)\). Then \(\vect{X}\) is reducible if and only if the set of cones \(\mathcal{K}=\{\vect{K}_1, \dots, \vect{K}_r\}\) has a complete extraction. \label{TheoremEquivalentReducibilityCondition}
\end{restatable}

The theorem is proved in Section \ref{subsec:ProofofEquivalentReducibilityCondidion}.

Once the theorem is proved, the algorithm for checking if an almost hybridlinear set $\vect{X}$ is reducible first computes those cones, and afterwards searches for a complete extraction. For correctness, the algorithm also relies upon the following lemma:

\begin{lemma}
Let \(\vect{Q}\) be a full periodic set. Then there exists a smooth periodic set \(\vect{P}'\) such that \(\Fill(\vect{P}')=\vect{Q}\) and \(\dir(\vect{P}')=\interior(\Q_{\geq 0} \vect{Q})\). \label{LemmaCanRequireInteriorAsDirections}
\end{lemma}

\begin{proof}
We prove this in the special case of \(\vect{Q}=\N^n\), via an appropriate linear map we then obtain the result for all full periodic sets whose cone has \(\dim(\vect{Q})\) many generators, and can extend to the general case.

Hence let \(\vect{Q}=\N^n\). Define \(\vect{P}':=\{(x_1, \dots, x_n) \mid x_i \leq 2^{x_j}-1 \forall i \neq j\}\). First of all, let \(\vect{v}=(10, \dots, 10)\) be the vector with all components set to \(10\). Then \(\vect{v}+\vect{e}_i \in \vect{P}'\) for all unit vectors \(\vect{e}_i\), and hence \(\vect{e}_i \in \vect{P}'-\vect{P}'\). Considering the asymptotics, we have that \(\overline{\Q_{\geq 0}\vect{P}'}=\Q_{\geq 0}^n\). Hence \(\Fill(\vect{P}')=\N^n\). Since every vector in \(\partial(\Q_{\geq 0}^n)\) leaves some component unchanged, none of those can be directions of \(\vect{P}'\). Hence \(\dir(\vect{P}')=\interior(\Q_{\geq 0}^n)\) using Lemma \ref{LemmaDirectionsContainInterior}. \(\dir(\vect{P}')\) is definable by definition, and well-directed is similarly obvious, since any line parallel to the boundary can only contain finitely many points of \(\vect{P}'\), i.e. any infinite sequence \((\vect{p}_m)_m \subseteq \vect{P}'\) has to contain a subsequence with \(p_m - p_k \in \interior(\Q_{\geq 0}^n)=\dir(\vect{P}')\) for all \(m>k\). Hence we have found our choice of \(\vect{P}'\).
\end{proof}

\TheoremReducibilityIsDecidable*

\begin{proof}
Write \(\vect{X}:=\vect{R} \cap \vect{S}\). The algorithm and its proof are split into two parts: First obtain a representation \(\vect{X}=\bigcup_{i=1}^r \vect{b}_i + \vect{P}_i\) with \(\Fill(\vect{P}_i)=\Fill(\vect{P}_j)\) for all \(i,j\), and then check whether \(\{\dir(P_i) \mid 1 \leq i \leq r\}\) has a complete extraction.

By Proposition \ref{PropositionPropertiesOfHybridization}(5), we can decompose \(\vect{X}=\vect{X}_1 \cup \dots \cup \vect{X}_r\) via KLMST-decomposition. We only used that the full linear hybridizations can be computed, but in fact even more is true: In \cite{Hauschildt90}, Hauschildt shows that whether a given vector \(\vect{d}\) is an element of \(\dir(\vect{X}_i)\) can be decided. Though we will not explain this in detail, it can be upgraded to compute a representation of the cone \(\dir(\vect{X}_i)\). Instead we deal with the second problem: This representation \(\vect{X}=\vect{X}_1 \cup \dots \cup \vect{X}_r\) might not be an almost hybridlinear representation, i.e. the fills might differ. 

Write \(\vect{S}=\vect{c}+\vect{Q}\). The algorithmic solution to problem 2 is simple: Add \(\interior(\vect{Q})\) to all cones \(\dir(\vect{X}_i)\), i.e. consider \(\mathcal{K}=\{\vect{K}_i:=\dir(\vect{X}_i)+\interior(\vect{Q}) \mid 1 \leq i \leq r\}\).

Check for a complete extraction of \(\mathcal{K}\) using two semi-algorithms: One to check whether the set of cones \(\mathcal{K}\) does not fulfil the property of Lemma \ref{CharOfCompleteExtraction}, and one that searches for a complete extraction. Output the answer of the semi-algorithm which terminates.

\smallskip\noindent Termination: By Lemma \ref{CharOfCompleteExtraction}.

\smallskip\noindent Correctness: Let \(\vect{P}\) smooth such that \(\vect{X}+\vect{P} \subseteq \vect{X}\) and \(\Fill(\vect{P})=\vect{Q}\). Let \(\vect{P}'\) be the smooth periodic as in Lemma \ref{LemmaCanRequireInteriorAsDirections}. Then \(\vect{P}'':=\vect{P} \cap \vect{P}'\) is smooth with \(\Fill(\vect{P}'')=\vect{Q}\) and \(\dir(\vect{P}'')=\interior(\Q_{\geq 0}\vect{Q})\) by Proposition \ref{PropositionIntersectionPeriodicSets}. Furthermore, we have \(\vect{X}+\vect{P}'' \subseteq \vect{X}\). Inspecting the proof of Theorem \ref{TheoremEquivalentAlmostHybridlinearCondition}, \(\vect{X}\) has an almost hybridlinear representation with periodic sets \(\vect{P}_i + \vect{P}''\). By Lemma \ref{LemmaSumPreservesSmooth}, these fulfil \(\dir(\vect{P}_i + \vect{P}'')=\dir(\vect{P}_i)+\dir(\vect{P}'')\), i.e. there is an almost hybridlinear representation with the cones considered by the algorithm. Correctness then follows by Theorem \ref{TheoremEquivalentReducibilityCondition}.
\end{proof}

\subsection{Proof of Theorem \ref{TheoremEquivalentReducibilityCondition}} 
\label{subsec:ProofofEquivalentReducibilityCondidion}

We require two geometric properties of almost hybridlinear sets. We start with an auxiliary lemma.

\begin{lemma}
Let \(\vect{P}\) be a periodic set, and \(\vect{F} \subseteq \vect{P}-\vect{P}\) finite. Then there exists \(\vect{p} \in \vect{P}\) such that \(\vect{p}+\vect{F} \subseteq \vect{P}\). \label{LemmaFinitePumping}
\end{lemma}

\begin{proof}
Write \(\vect{F}=\{\vect{p}_1, \dots, \vect{p}_s\}\). Write \(\vect{p}_i=\vect{p}_{i,1}-\vect{p}_{i,2}\) with \(\vect{p}_{i,1}, \vect{p}_{i,2} \in \vect{P}\) for every \(i\). Define \(\vect{p}:=\sum_{i=1}^s \vect{p}_{i,2}\). Then \(\vect{p}+\vect{p}_i \in \vect{P}\) for all \(i\).
\end{proof}

The first property essentially allows us to ignore the lattice and only consider cones.

\begin{restatable}{proposition}{PropositionModuloIssuesRemoved}
Let \(\vect{X}\) be a set with weak hybridization \(\vect{B}+\Fill(\vect{P})\). Let \(\vect{Q} \subseteq \Fill(\vect{P})\) be a finitely generated periodic set. If \(\vect{x}+\lambda \vect{Q} \subseteq \vect{X}\) for some \(\vect{x} \in \vect{X}, \lambda \in \N_{>0}\), then \(\vect{x}' + \vect{Q} \subseteq \vect{X}\) for some \(\vect{x}' \in \vect{X}\). \label{PropositionModuloIssuesRemoved}
\end{restatable}

\begin{proof}
Write \(\vect{Q}=\{\vect{d}_1', \dots, \vect{d}_s'\}^{\ast}\) and assume that \(\vect{x}+\lambda \{\vect{d}_1', \dots, \vect{d}_s'\}^{\ast} \subseteq \vect{X}\). Define \(\vect{F}:=\{0, \dots, \lambda -1\}\vect{d}_1' + \dots + \{0, \dots, \lambda -1\}\vect{d}_s' \subseteq \Fill(\vect{P}) \subseteq \vect{P}-\vect{P}\). By Lemma \ref{LemmaFinitePumping} there exists \(\vect{d} \in \vect{P}\) such that \(\vect{d}+\vect{F} \subseteq \vect{P}\). Choose \(\vect{x}':=\vect{x}+\vect{d}\). Then we claim \(\vect{x}' + \vect{Q} \subseteq \vect{X}\). To see this, let \(\vect{y}=\vect{x}' + \lambda_1 \vect{d}_1' + \dots + \lambda_s \vect{d}_s' \in \vect{x}' + \vect{Q}\). Let \(\lambda_i':=\lambda_i \mod \lambda\), and observe that \(\vect{w}':=\lambda_1' \vect{d}_1' + \dots + \lambda_s' \vect{d}_s' \in \vect{F}\), and \(\vect{y}-\vect{w}' \in \vect{x}+\vect{d}+\lambda \vect{Q}\). In total we obtain 
\begin{align*}
\vect{y}&=\vect{w}'+(\vect{y}-\vect{w}')\in \vect{F}+\vect{x}+\vect{d}+\lambda \vect{Q}=(\vect{x}+\lambda \vect{Q})+(\vect{d}+\vect{F}) \subseteq \vect{X}+\vect{P} \subseteq \vect{X}.
\end{align*}

\end{proof}

Next the second property. Intuitively, we can move starting points of lines, planes, etc.\ together.

\begin{proposition}
Let \(\vect{X}\) be a set with hybridization \(\vect{c}+\Fill(\vect{P})\). Assume that there exist finitely many sets \(\vect{G}_1, \dots, \vect{G}_s \subseteq \Fill(\vect{P})\) and points \(\vect{x}_1, \dots, \vect{x}_s \in \vect{X}\) such that \(\vect{x}_i+\vect{G}_i \subseteq \vect{X}\) for all \(i\). Then there exists \(\vect{x}'\) such that \(\vect{x}'+(\bigcup_{i=1}^s \vect{G}_i) \subseteq \vect{X}\). \label{PropositionOnePointWithCombinedPeriodicityPower}
\end{proposition}

\begin{proof}
Define \(\vect{F}:=\{\vect{c}-\vect{x}_1, \dots, \vect{c}-\vect{x}_s\} \subseteq -\Fill(\vect{P}) \subseteq \vect{P}-\vect{P}\). By Lemma \ref{LemmaFinitePumping}, there exists \(\vect{p}\in \vect{P}\) such that \(\vect{p}+\vect{F} \subseteq \vect{P}\). We define \(\vect{x}':=\vect{c}+\vect{p}\), and claim that \(\vect{x}'+(\bigcup_{i=1}^s \vect{G}_i) \subseteq \vect{X}\). To see this, let \(\vect{x}'+\vect{y} \in \vect{x}' + (\bigcup_{i=1}^s \vect{G}_i)\). We have \(\vect{y}\in \vect{G}_i\) for some \(i\), and obtain the required
\begin{align*}
\vect{x}'+\vect{y}&=(\vect{c}-\vect{x}_i+\vect{x}_i+\vect{p})+\vect{y}=(\vect{x}_i+\vect{y})+(\vect{p}+(\vect{c}-\vect{x}_i)) \in \vect{X} + \vect{P} \subseteq \vect{X}.
\end{align*}
\end{proof}

The main use case of Proposition \ref{PropositionModuloIssuesRemoved} is obtained from the following lemma.

\begin{lemma}
Let \(\vect{Q},\vect{Q}'\) be periodic sets with the same finitely generated cone \(\vect{C}=\mathbb{Q}_{\geq 0}\vect{Q}=\mathbb{Q}_{\geq 0}\vect{Q}'\). Then \(\lambda \cdot \vect{Q} \subseteq \vect{Q}'\) for some \(\lambda \in \N_{>0}\). \label{LemmaSameConeThenFactorLambda}
\end{lemma}

\begin{proof}
By Lemma \ref{LemmaCharactizeFinitelyGeneratedPeriodicSets}, both \(\vect{Q}\) and \(\vect{Q}'\) are finitely generated. Write \(\vect{Q}=\{\vect{p}_1, \dots, \vect{p}_s\}^{\ast}\). Since \(\vect{p}_i \in \mathbb{Q}_{\geq 0}\vect{Q}'\), there exists \(\lambda_i \in \N_{>0}\) such that \(\lambda_i \cdot \vect{p}_i \in \vect{Q}'\). It follows that \(\lambda:=\prod_{i=1}^s \lambda_i\) fulfills \(\lambda \cdot \vect{Q} \subseteq \vect{Q}'\).
\end{proof}

Now we are finally ready to prove Theorem \ref{TheoremEquivalentReducibilityCondition}. 

\TheoremEquivalentReducibilityCondition*

\begin{proof}
``\(\Leftarrow\)'': Let \(\vect{C}_1, \dots, \vect{C_r}\) be a complete extraction of \(\mathcal{K}\). We first claim that \(\bigcup_{i=1}^r \vect{K}_i=\mathbb{Q}_{\geq 0}\Fill(\vect{P})\).

Proof of claim: \(\subseteq\) is clear, for the other direction first use Lemma \ref{LemmaDirectionsContainInterior} to obtain that \(\interior(\mathbb{Q}_{\geq 0}\Fill(\vect{P}))=\interior(\overline{\mathbb{Q}_{\geq 0}\vect{P}_i}) \subseteq \dir(\vect{P}_i)=\vect{K}_i\). Since every \(\vect{C}_i\) is finitely generated and hence closed, we have that \(\bigcup_{i=1}^r \vect{K}_i=\bigcup_{i=1}^r \vect{C}_i\) is closed. Therefore we obtain that \(\mathbb{Q}_{\geq 0}\Fill(\vect{P})=\overline{\interior(\mathbb{Q}_{\geq 0}\Fill(\vect{P}))} \subseteq \bigcup_{i=1}^r \vect{K}_i\) as claimed.

The idea for the rest is as follows: If \(\vect{F}_i=\{\vect{d}_1, \dots, \vect{d}_s\}\) is a set of directions of a periodic set \(\vect{P}_i\), then \(\vect{x}_i+\lambda_i \vect{F}_i^{\ast} \subseteq \vect{P}_i\) for some starting point \(\vect{x}_i\). We will remove the factor \(\lambda_i\), and then use use \(\vect{F}_i^{\ast}\) as \(\vect{G}_i\) as in Proposition \ref{PropositionOnePointWithCombinedPeriodicityPower} to finish the proof. 

Formally: For every \(i \in \{1, \dots, r\}\), we do the following. Let \(\vect{F}_i\) be a finite set of generators of \(\vect{C}_i\). Since \(\vect{F}_i \subseteq \VectorSpace(\vect{P_i})=\Q_{\geq 0}(\vect{P}-\vect{P})\), by replacing \(\vect{F}_i\) by multiples we can assume \(\vect{F}_i \subseteq \vect{P}_i-\vect{P}_i\). Then by Proposition \ref{PropositionAlmostPeriodicityPeriodicSet}, there exists \(\vect{x}_i'\) such that \(\vect{x}_i'+\vect{F}_i^{\ast} \subseteq \vect{P}_i\) and hence \((\vect{b}_i+\vect{x}_i') + \vect{F}_i^{\ast} \subseteq \vect{b}_i+\vect{P}_i \subseteq \vect{X}\). 

Observe that \(\vect{F}_i^{\ast}\) is a finitely generated periodic set with the same cone as \(\vect{C}_i \cap \Fill(\vect{P})\). Hence by Lemma \ref{LemmaSameConeThenFactorLambda} there exists \(\lambda_i\) such that \(\lambda_i \cdot (\vect{C}_i \cap \Fill(\vect{P})) \subseteq \vect{F}_i^{\ast}\). By Proposition \ref{PropositionModuloIssuesRemoved}, there exists \(\vect{x}_i\) such that \(\vect{x}_i+(\vect{C}_i \cap \Fill(\vect{P})) \subseteq \vect{X}\). By Proposition \ref{PropositionOnePointWithCombinedPeriodicityPower}, there exists \(\vect{x}\) such that \(\vect{x}+\bigcup_{i=1}^r (\vect{C}_i \cap \Fill(\vect{P})) \subseteq \vect{X}\). Since \(\bigcup_{i=1}^r (\vect{C}_i \cap \Fill(P))=\left (\bigcup_{i=1}^r \vect{C}_i \right) \cap \Fill(\vect{P})=\mathbb{Q}_{\geq 0}\Fill(\vect{P}) \cap \Fill(\vect{P})=\Fill(\vect{P})\), we have \(\vect{x}+\Fill(\vect{P})\subseteq \vect{X}\) and hence \(\vect{X}\) is reducible.

``\(\Rightarrow\)'': This direction follows from \cite[Lemma F.5, F.6]{Leroux13}. Their argument was slightly more involved because they did not assume that all \(\vect{P}_i\) define the same lattice \(\vect{P}'-\vect{P}'\), accordingly they state that all \(\mathcal{K}_{V,z}\) fulfill the property of Lemma \ref{CharOfCompleteExtraction}, i.e. have a complete extraction. In our case there is exactly one \(\mathcal{K}_{V,z}\) and that is \(\mathcal{K}\).
\end{proof}

\section{Proofs of Section \ref{SectionFinalPartitionAndCorollaries}}

\subsection{Proofs of Section \ref{subsec:SemilinearityIsDecidbale}}

\LemmaSemilinearThenReducible*

\begin{proof}
To complete the proof idea of Section \ref{subsec:SemilinearityIsDecidbale}, we need to argue that our intuitive reasoning of every limit being attained is correct, and that this actually implies reducibility. We start with the second part.

Since \(\vect{X}\) is semilinear, we have \(\vect{X}=\bigcup_{i=1}^r \vect{b}_i+\vect{P}_i\) with full periodic sets \(\vect{P}_i\). In particular the cones \(\mathbb{Q}_{\geq 0}\vect{P}_i\) are finitely generated. We cannot simply use Theorem \ref{TheoremEquivalentReducibilityCondition} immediately, since the semilinear representation will almost definitely not fulfill \(\Fill(\vect{P}_1)=\dots=\Fill(\vect{P}_r)\). Instead, as in the proof of Theorem \ref{TheoremEquivalentAlmostHybridlinearCondition}, we have \(\vect{X}=\bigcup_{i=1}^r \vect{b}_i+(\vect{P}_i+\vect{P})\), and the \(\vect{P}_i':=\vect{P}_i+\vect{P}\) fulfill \(\Fill(\vect{P}_i')=\Fill(\vect{P})\). By Theorem \ref{TheoremEquivalentReducibilityCondition} it hence suffices to show that \(\{\dir(\vect{P}_1'), \dots, \dir(\vect{P}_r')\}\) has a complete extraction. We do this by showing \(\bigcup_{i=1}^r \mathbb{Q}_{\geq 0}\vect{P}_i=\mathbb{Q}_{\geq 0}\Fill(\vect{P})\), at which point the cones \(\vect{C}_i:=\mathbb{Q}_{\geq 0}\vect{P}_i \subseteq \dir(\vect{P}_i')\) form a complete extraction of the required set. This claim about the cones basically corresponds to the intuition of ``every direction is attained''.

Claim 1: \(\Q_{\geq 0}\vect{P} \subseteq \bigcup_{i=1}^r \Q_{\geq 0}\vect{P}_i\). 

Proof of claim 1: Let \(\vect{p} \in \vect{P}\). Since \(\vect{X} \neq \emptyset\), there exists \(\vect{x} \in \vect{X}\). Since \(\vect{x} \in \vect{X}\) and \(\vect{X}+\vect{P} \subseteq \vect{X}\), the sequence \((\vect{x}+\lambda \vect{p})_{\lambda \in \N} \subseteq \vect{X}\). Hence all of these points are in some \(\vect{c}_i+\vect{P}_i\). By pigeonhole principle, some \(\vect{P}_i\) contains infinitely many, in particular some \(\lambda_1 \vect{p} \in \vect{P}_i \) and \(\lambda_2 \vect{p} \in \vect{P}_i\). Then \((\lambda_1-\lambda_2) \vect{p} \in \vect{P}_i - \vect{P}_i\), and furthermore \(\vect{p} \in \overline{\mathbb{Q}_{\geq 0}\vect{P}_i}\), since infinitely many elements from the sequence are contained in \(\vect{P}_i\). In total \((\lambda_1-\lambda_2)\vect{p} \in \Fill(\vect{P}_i)=\vect{P}_i\), since \(\vect{P}_i\) is full. Then \(\vect{p} \in \Q_{\geq 0}\vect{P}_i\) as claimed.

Claim 2: \(\overline{\mathbb{Q}_{\geq 0}\vect{P}}=\bigcup_{i=1}^r \mathbb{Q}_{\geq 0}\vect{P}_i\), which would finish the proof by observing \(\overline{\mathbb{Q}_{\geq 0}\vect{P}}=\mathbb{Q}_{\geq 0}\Fill(\vect{P})\). 

Proof of claim 2: ``\(\supseteq\)'' follows from Lemma \ref{LemmaFullThenLessPeriods}. Hence let \(d \in \overline{\mathbb{Q}_{\geq 0}\vect{P}}\). Then there exists a sequence \((\vect{d}_n)_{n\in \N} \subseteq \mathbb{Q}_{\geq 0}\vect{P}\) converging to \(\vect{d}\). By Claim 1, \(\mathbb{Q}_{\geq 0}\vect{P} \subseteq \bigcup_{i=1}^r \mathbb{Q}_{\geq 0}\vect{P}_i\), hence infinitely many of the \(\vect{d}_n\) are in the same \(\mathbb{Q}_{\geq 0}\vect{P}_i\), and we obtain \(\vect{d}\in \overline{\mathbb{Q}_{\geq 0}\vect{P}_i}\) for some \(i\). Since \(\vect{P}_i\) is finitely generated, \(\mathbb{Q}_{\geq 0}\vect{P}_i\) is closed, and hence \(\vect{d}\in \mathbb{Q}_{\geq 0}\vect{P}_i\).
\end{proof}

\subsection{Proofs of Section \ref{SubsectionInfiniteLineCorollary}}

We want to follow the intuition depicted in the Figure in Section \ref{SubsectionInfiniteLineCorollary}. The main difficulty is to define a ``broad enough'' cone \(\vect{C}\), then we simply use Proposition \ref{PropositionAlmostPeriodicityPeriodicSet} on the finite set \(\vect{F}\) generating the full periodic set \(\vect{C} \cap (\vect{P}-\vect{P})\) to obtain \(\vect{X}+\vect{v}+\vect{C} \subseteq \vect{X}\). Furthermore, in order to ensure that \(\vect{F}\) contains directions, we may only choose interior vectors for \(\vect{F}\). For example, if the hybridization is \(\N^n\), then \(\vect{F}\) will only contain vectors with every coordinate strictly positive. In general, for a cone given via hyperplanes, we need to define a distance from those hyperplanes.

\begin{lemma}
Let \(\vect{Q}\) be a full periodic set. Then there exists a finite set \(\vect{F} \subseteq \interior(\vect{Q})\) such that \(\vect{Q} \subseteq \partial(\vect{Q})+\vect{F}^{\ast}\). \label{LemmaFillOutTheBoundary}
\end{lemma}

\begin{proof}
Setup: A relation \(\leq\) on \(\vect{X}\) is called a quasi-order if it is reflexive and transitive. A preorder is called well-quasi-order if every upward-closed subset of \(\vect{X}\) has a finite basis, i.e. finitely many minimal elements. By \cite[Lemma V.5]{Leroux13}, if a periodic set \(\vect{Q}\) is finitely generated, then \(\vect{Q}\) is well-preordered by \(\leq_{\vect{Q}}\) defined via \(\vect{x} \leq_{\vect{Q}} \vect{y} \iff \vect{y}-\vect{x} \in \vect{Q}\). Hence in particular the set \(\interior(\vect{Q})\) has a finite set of minimal elements \(\vect{F}\) w.r.t. \(\leq_{\vect{Q}}\).

Proof of lemma: Since \(\vect{C}:=\Q_{\geq 0} \vect{Q}\) is finitely generated, by Lemma \ref{LemmaFinitelyGeneratedCones} there exists an integer matrix \(A\) such that \(\vect{C}=\{\vect{x} \in \VectorSpace(Q) \mid A \vect{x} \geq \vect{0}\}\), and the faces are \(\vect{C}_i:=\{\vect{x} \in \vect{C} \mid A_i \vect{x} = 0\}\) for the row \(A_i\). 

Since \(\vect{Q}\subseteq \N^n\), we have \(A_i \vect{x} \in \N\) for all \(\vect{x}\in \vect{Q}\). To every point \(\vect{x}\in \vect{Q}\), we can hence assign \(\invariant(\vect{x}):=\min_i (A_i \vect{x}) \in \N\). This measures distance to the closest boundary. The proof is by induction on this distance.

If \(\invariant(\vect{x})=0\), then \(\vect{x}\in \vect{C}_i\) for some \(i\), and hence \(\vect{x}\in \partial(\vect{Q})\) as claimed.

Otherwise \(\vect{x}\in \interior(\vect{Q})\), and hence by definition of \(\leq_{\vect{Q}}\) and \(\vect{F}\), there exist \(\vect{f}\in \vect{F}\) and \(\vect{q}\in \vect{Q}\) such that \(\vect{x}=\vect{f}+\vect{q}\). Since \(\vect{f}\in \interior(\vect{Q})\), we have \(A_i \vect{f} >0\) for every \(i\). Hence \(\invariant(\vect{q})<\invariant(\vect{x})\). By induction, \(\vect{q}\in \partial(\vect{Q})+\vect{F}^{\ast}\), and hence \(\vect{x}=\vect{q}+\vect{f} \in \partial(\vect{Q})+\vect{F}^{\ast}\).
\end{proof}

\PropositionBoundaryImpliesCone*

\begin{proof}
Since \(\vect{X}\) contains almost the whole boundary of \(\vect{c}+\Fill(\vect{P})\), in particular for every facet \(\vect{G}_i\) of \(\Fill(\vect{P})\), there exists an \(\vect{x}_i\) such that \(\vect{x}_i+\vect{G}_i \subseteq \vect{X}\). By Proposition \ref{PropositionOnePointWithCombinedPeriodicityPower}, there exists \(\vect{x}'\) such that \(\vect{x}'+\partial(\Fill(\vect{P})) \subseteq \vect{X}\). By Lemma \ref{LemmaFillOutTheBoundary}, there exists \(\vect{F} \subseteq \interior(\Fill(\vect{P}))\) such that \(\Fill(\vect{P}) \subseteq \partial(\Fill(\vect{P}))+\vect{F}^{\ast}\). By Lemma \ref{LemmaDirectionsContainInterior}, we have \(\interior(\Fill(\vect{P})) \subseteq \dir(\vect{P})\). We furthermore have \(\vect{F} \subseteq (\vect{P}-\vect{P})\) by definition. Hence, by Proposition \ref{PropositionAlmostPeriodicityPeriodicSet}, there exists a \(\vect{d}\) such that \(\vect{d}+\vect{F}^{\ast} \subseteq \vect{P}\). We define \(\vect{x}:=\vect{x}'+\vect{d}\), and obtain
\begin{align*}
\vect{x}+\Fill(\vect{P})&\subseteq (\vect{x}'+\vect{d})+(\partial(\Fill(\vect{P}))+\vect{F}^{\ast}) \\
&=(\vect{x}'+\partial(\Fill(\vect{P})))+(\vect{d}+\vect{F}^{\ast})\subseteq \vect{X}+\vect{P} \subseteq \vect{X}.
\end{align*}
\end{proof}

\subsection{Separating a target Petri set}

We show that if a VAS reachability set does not intersect a target Petri set, then there exists a semilinear inductive invariant separating them. We start by proving that for two given Petri sets \(\vect{X}_1\) and \(\vect{X}_2\) and a semilinear set \(\vect{S}\), there is a common partition \(\vect{S}=\vect{S}_1 \cup \dots \cup \vect{S}_k\) which fulfills the conditions of Theorem \ref{TheoremFinalPartition} with respect to both \(\vect{X}_1\) and \(\vect{X}_2\).

\begin{corollary}
Let \(\vect{X}_1\), \(\vect{X}_2\) be Petri sets. For every semilinear set \(\vect{S}\) there exists a partition \(\vect{S}=\vect{S}_1 \cup \dots \cup \vect{S}_k\) into pairwise disjoint full linear sets such that for all \(i \in \{1, \dots, k\}\) and \(j\in \{1,2\}\) either \(\vect{X}_j \cap \vect{S}_i=\emptyset\), \(\vect{S}_i \subseteq \vect{X}_j\) or \(\vect{X}_j \cap \vect{S}_i\) is an irreducible almost hybridlinear set with hybridization $\vect{S}_i$. Further, if \(\vect{X}_1\) and \(\vect{X}_2\) are reachability sets of VASs \(\mathcal{V}_1\) and \(\mathcal{V}_2\), then the partition is computable. \label{CorollaryCommonPartition}
\end{corollary}

\begin{proof}
The following procedure computes such a partition.

Step 1: Use Theorem \ref{TheoremFinalPartition} with \(\vect{X}=\vect{X}_1\) and \(\vect{S}=\vect{S}\) to compute a partition \(\vect{S}=\vect{S}_1 \cup \dots \cup \vect{S}_r\) fulfilling the properties for \(\vect{X}_1\). For every \(i\), we compute a subpartition of \(\vect{S}_i\) as follows.

Step 2: Use Theorem \ref{TheoremFinalPartition} with \(\vect{X}=\vect{X}_2\) and \(\vect{S}=\vect{S}_i\) to compute a partition \(\vect{S}_i=\bigcup_{j=1}^{k_i} \vect{S}_{i,j}\) fulfilling the properties for \(\vect{X}_2\). If \(\vect{X}_1 \cap \vect{S}_i=\emptyset\) or \(\vect{S}_i \subseteq \vect{X}_1\), then end step 2.

Otherwise \(\vect{X}_1 \cap \vect{S}_i\) is irreducible almost hybridlinear. For every \(j\) do the following:

Case 1: \(\dim(\vect{S}_{i,j})<\dim(\vect{S}_i)\): Perform a recursive call with \(\vect{X}_1, \vect{X}_2\) and \(\vect{S}=\vect{S}_{i,j}\) to obtain an appropriate partition of \(\vect{S}_{i,j}\).

Case 2: \(\dim(\vect{S}_{i,j})=\dim(\vect{S}_i)\): Decide whether \(\vect{X}_1 \cap \vect{S}_{i,j}\) is reducible, and whether \(\vect{X}_1 \cap \vect{S}_{i,j}=\emptyset\):

Case 2.1: Irreducible or empty: Then leave \(\vect{S}_{i,j}\) as is.

Case 2.2: \(\dim(\vect{S}_{i,j})=\dim(\vect{S}_i)\) and \(\vect{X}_1 \cap \vect{S}_{i,j}\) is reducible (possibly, but not necessarily entire $\vect{S}_{i,j}$): Then write \(\vect{S}=\vect{c}+\vect{Q}\) and find \(\vect{x}\) such that \(\vect{x}+\vect{Q} \subseteq \vect{X}_1 \cap \vect{S}_{i,j}\). Afterwards do a recursive call with \(\vect{X}_1=\vect{X}_1\), \(\vect{X}_2=\vect{X}_2\) and \(\vect{S}=\vect{S}_{i,j} \setminus (\vect{x}+\vect{Q})\), and use Theorem \ref{TheoremFinalPartition} with \(\vect{X}=\vect{X}_2\) and \(\vect{S}=\vect{x}+\vect{Q}\) to obtain a partition of \(\vect{x}+\vect{Q}\). Combine the two partitions.

Set the partition of \(\vect{S}_i\) to the union of the partitions of all the \(\vect{S}_{i,j}\).

Step 3: Now simply return the union of the subpartitions for all the \(\vect{S}_i\).

Termination: By Lemma \ref{LemmaRemovingInnerConeReducesDimension}, we only perform recursion on sets \(\vect{S}'\) with \(\dim(\vect{S}')<\dim(\vect{S})\). Hence recursion depth is at most \(\dim(\vect{S})\) and termination immediate.

Correctness: In case the ``Otherwise'' does not occur: For \(\vect{X}_2\) the properties follow from Theorem \ref{TheoremFinalPartition}. For \(\vect{X}_1\), we have that \(\vect{X}_1 \cap \vect{S}_{i,j}\) is still empty or \(\vect{S}_{i,j} \subseteq \vect{S}_i \subseteq \vect{X}_1\).

Case 1: Correct by induction/recursion.

Case 2: By Proposition \ref{PropositionPropertiesOfHybridization}(2), \(\vect{X}_1 \cap S_{i,j}\) is almost hybridlinear with hybridization \(\vect{S}_{i,j}\), hence reducibility is defined.

Case 2.1: By definition of the case, we have \(\vect{X}_1 \cap \vect{S}_{i,j}\) is either empty or irreducible almost hybridlinear with hybridization \(\vect{S}_{i,j}\). For \(\vect{X}_2\), correctness follows from Theorem \ref{TheoremFinalPartition}.

Case 2.2: For the partition parts \(\vect{S}_{i,j,k}\) of \(\vect{x}+\vect{Q}\), we have \(\vect{S}_{i,j,k} \subseteq \vect{X}_1\). For \(\vect{X}_2\), the properties hold by correctness of Theorem \ref{TheoremFinalPartition}. For the partition parts \(\vect{S}_{i,j,k}'\) of \(\vect{S}_{i,j} \setminus (\vect{x}+\vect{Q})\), we have correctness by induction.
\end{proof}

\CorollarySeparatingPairOfPetriSets*

\begin{proof}
Let \(\vect{S}=\N^n\). Let \(\vect{S}=\vect{S}_1 \cup \dots \cup \vect{S}_k\) be the partition of Corollary \ref{CorollaryCommonPartition}. Let \(I \subseteq \{1, \dots, k\}\) be the set of indices such that \(\vect{X}_1 \cap \vect{S}_i \neq \emptyset\). We claim that \(\vect{S}'=\bigcup_{i\in I} \vect{S}_i\) fulfills the result. \(\vect{X}_1 \subseteq \vect{S}'\) is obvious by construction. Hence let \(i \in I\). We are going to show that \(\vect{S}_i \cap \vect{X}_2=\emptyset\).

If \(i\) is an index with \(\vect{S}_i \subseteq \vect{X}_1\), then we clearly have \(\vect{X}_2 \cap \vect{S}_i=\emptyset\) since \(\vect{X}_1 \cap \vect{X}_2=\emptyset\).

Hence let \(i\) be an index such that \(\vect{X}_1 \cap \vect{S}_i\) is almost hybridlinear with hybridization \(\vect{S}_i\). We will prove \(\vect{X}_2 \cap \vect{S}_i=\emptyset\) by contradicting the other cases. 

If \(\vect{S}_i \subseteq \vect{X}_2\), then we have a contradiction to \(\vect{X}_1 \cap \vect{X}_2 =\emptyset\), since \(\vect{X}_1 \cap \vect{S}_i \neq \emptyset\) by definition of almost hybridlinear.

Now assume for contradiction that \(\vect{X}_2 \cap \vect{S}_i\) is almost hybridlinear with hybridization \(\vect{S}_i\). Let \(\vect{x}_1 \in \vect{X}_1 \cap \vect{S}_i\) and \(\vect{x}_2 \in \vect{X}_2 \cap \vect{S}_i\). Write \(\vect{S}_i=\vect{c}+\vect{Q}\). Then let \(\vect{P}_1, \vect{P}_2\) smooth such that \(\vect{X}_1+\vect{P}_1 \subseteq \vect{X}_1\), \(\vect{X}_2+\vect{P}_2 \subseteq \vect{X}_2\) and \(\Fill(\vect{P}_1)=\vect{Q}=\Fill(\vect{P}_2)\). By Proposition \ref{PropositionIntersectionPeriodicSets}, \(\vect{P}:=\vect{P}_1 \cap \vect{P}_2\) is smooth and \(\Fill(\vect{P})=\vect{Q}\). Similar to the proof of Proposition \ref{PropositionOnePointWithCombinedPeriodicityPower}, there exists \(\vect{x}'\) such that \(\vect{x}'-\vect{x}_i \in \vect{P}\) for both \(i\). Hence \(\vect{x}'=\vect{x}_i+(\vect{x}'-\vect{x}_i) \in \vect{X}_i + \vect{P}_i \subseteq \vect{X}_i\) for both \(i\), contradiction to \(\vect{X}_1 \cap \vect{X}_2=\emptyset\).

Therefore \(\vect{S}_i \cap \vect{X}_2=\emptyset\) is the only possibility left, and \(\vect{S} \cap \vect{X}_2\) is empty as claimed.
\end{proof}

\CorollarySeparatingPetriSet*

\begin{proof}
Let \(\vect{X}_1:=\Reach(\mathcal{V})\). By Theorem \ref{TheoremVASPetriSets}, \(\vect{X}_1\) is a Petri set. By Corollary \ref{CorollarySeparatingPetriSets}, there exists a semilinear set \(\vect{S}\) with \(\vect{X}_1 \subseteq \vect{S}\) and \(\vect{X}_2 \cap \vect{S}=\emptyset\). Let \(\vect{S}^C\) be the complement of \(\vect{S}\). By \cite{Leroux09}, there exists a semilinear inductive invariant \(\vect{S}'\) separating \(\vect{X}_1\) from \(\vect{S}^C\).
\end{proof}

\section{Example that an almost-linear partition does not exist in general} \label{SectionNoAlmostLinearPartition}

In this section we give an example that there exists a Petri set \(\vect{X}\), in fact we provide a VAS definable set \(\vect{X}\), such that there does not exist a partition \(\N^n=\vect{S}_1 \cup \dots \cup \vect{S}_r\) with the property that \(\vect{X} \cap \vect{S}_i\) is either empty or almost linear for all \(i\).

The example is 3-dimensional and defined as follows: Define \(\vect{P}_1:=\{(x,y,z) \in \N^3 \mid z \leq x \cdot y\}\), \(\vect{P}_2:=\{(x,y,z) \in \N^3 \mid y \leq 2^x, x \leq 2^y, x \leq 2^z\}\) and \(\vect{X}:=\vect{P}_1 \cup \vect{P}_2\). Since both \(\vect{P}_1\) and \(\vect{P}_2\) are smooth with \(\Fill(\vect{P}_i)=\N^3\), \(\vect{X}\) is almost-hybridlinear.

Let us first explain why we chose this example, followed by a proof sketch.

The reason for the example is that we have \(\dir(\vect{P}_1) \cup \dir(\vect{P}_2)=\Q_{\geq 0}^3\), i.e.\ \(\vect{X}\) contains a line in every direction, but \(\vect{X}\) is not reducible. That \(\vect{X}\) is not reducible follows from Theorem \ref{TheoremEquivalentReducibilityCondition}, or simply by observing that the points \((n, 2^{2^n}, 2^{2^{2^n}})\) are not in \(\vect{X}\). For the sets of directions, the important directions are boundary vectors, i.e. pumping vectors with at least one coordinate equal to \(0\). Regarding boundary vectors, \(\vect{P}_2\) is built to only allow \(z\) alone to be pumped, while \(\vect{P}_1\) allows any vector to be pumped which, if it pumps \(z\), also pumps either \(x\) or \(y\). I.e.\ the only impossible vector is \(z\) alone, which is pumpable in \(\vect{P}_2\).

In some sense, a set \(\vect{X}\) which is not reducible but does have every direction can be seen as a ``true'' irreducible almost hybridlinear set.

Proof sketch: Assume there exists a partition \(\N^3=\vect{S}_1 \cup \dots \cup \vect{S}_r\) such that \(\vect{X} \cap \vect{S}_i\) is almost linear for all \(i\). We will prove that this implies that \(\vect{X}\) is reducible and hence contradicts the above. The line of argument is as follows:

\begin{enumerate}
\item Prove that we can w.l.o.g. assume that \(\vect{S}_i\) is the fill of \(\vect{X} \cap \vect{S}_i\), where the fill of an almost linear set \(\vect{b}+ \vect{P}\) is defined as \(\vect{b}+\Fill(\vect{P})\).
\item Prove that \(\vect{X} \cap \vect{S}_i\) is reducible for all \(i\), where it is non-empty. This splits into two cases: \(\dim(\vect{S}_i)\leq 2\) and \(\dim(\vect{S}_i)=3\). In the first case we show that \(\vect{X} \cap \vect{S}_i\) is semilinear, and hence reducible by Lemma \ref{LemmaSemilinearThenReducible}, and in case 2 we will prove that \(\dir(\vect{X} \cap \vect{S}_i)\) is the whole cone of \(\vect{S}_i\), similar to how this holds for \(\vect{S}=\N^3\), the whole space, as mentioned above. As a side note, the fact that non-semilinearity of \(\vect{X}\) can only be shown here by considering three dimensional sections of \(\vect{X}\), i.e.\ that every intersection \(\vect{X} \cap \vect{S}\) for \(\dim(\vect{S}) \leq 2\) is semilinear, is interesting in and of itself.
\item Prove that if \(\vect{X} \cap \vect{S}_i\) is reducible for all \(i\) where it is non-empty, then \(\vect{X}\) is reducible.
\end{enumerate}

For geometric intuition, it is important to understand that everything we have to prove only has to do with directions, i.e.\ cones, and hence can be imagined one dimension lower, here in 2D. This allows us to give a visual explanation of the arguments in Figure \ref{FigureExampleAppendix}.

\begin{figure}[t]
\begin{minipage}{4.6cm}
\begin{tikzpicture}
\begin{axis}[
    xlabel={ },
    ylabel={ },
    xmin=0, xmax=4,
    ymin=0, ymax=4,
    xtick={0,1,2,3,4},
    ytick={0,1,2,3,4},
    ymajorgrids=true,
    xmajorgrids=true,
]

\addplot+[
    name path=AA,
    color=black,
    no marks,
    very thick,
    ]
    coordinates {
    (0.1,0.1)(2,3.4)
    };
    
\addplot+[
    name path=BA,
    color=black,
    no marks,
    very thick,
    ]
    coordinates {
    (0.1,0.1)(2,0.1)
    };
    
\addplot[
    color=blue!40,
]
fill between[of=AA and BA];

\addplot+[
    name path=AB,
    color=black,
    no marks,
    very thick,
    ]
    coordinates{
    (2,3.4)(3.9,0.1)
    };
    
\addplot+[
    name path=BB,
    color=black,
    no marks,
    very thick,
    ]
    coordinates {
    (2,0.1)(3.9,0.1)
    };
    
\addplot[
    color=blue!40,
    domain=1.5:4,
]
fill between[of=AB and BB];

\addplot[
    fill=red,
    fill opacity=1,
    only marks,
]
    coordinates{
    (0.1,0.1)
    };

\end{axis}
\end{tikzpicture}
\end{minipage}%
\begin{minipage}{4.6cm}
\begin{tikzpicture}
\begin{axis}[
    xlabel={ },
    ylabel={ },
    xmin=0, xmax=4,
    ymin=0, ymax=4,
    xtick={0,1,2,3,4},
    ytick={0,1,2,3,4},
    ymajorgrids=true,
    xmajorgrids=true,
]

\addplot+[
    name path=AA,
    color=black,
    no marks,
    very thick,
    ]
    coordinates {
    (0.1,0.1)(2,3.4)
    };
    
\addplot+[
    name path=BA,
    color=black,
    no marks,
    very thick,
    ]
    coordinates {
    (0.1,0.1)(2,0.1)
    };
    
\addplot[
    color=blue!40,
]
fill between[of=AA and BA];

\addplot+[
    name path=AB,
    color=black,
    no marks,
    very thick,
    ]
    coordinates{
    (2,3.4)(3.9,0.1)
    };
    
\addplot+[
    name path=BB,
    color=black,
    no marks,
    very thick,
    ]
    coordinates {
    (2,0.1)(3.9,0.1)
    };
    
\addplot[
    color=blue!40,
    domain=1.5:4,
]
fill between[of=AB and BB];

\addplot[
    color=black,
    no marks,
    very thick,
    ]
    coordinates {
    (1,0.1)(1,1.7)
    };
    
    \addplot[
    color=black,
    no marks,
    very thick,
    ]
    coordinates {
    (2,0.1)(3,1.7)
    };
    
    \addplot[
    color=black,
    no marks,
    very thick,
    ]
    coordinates {
    (1,0.1)(2,3.4)
    };
    
    \addplot[
    color=black,
    no marks,
    very thick,
    ]
    coordinates {
    (1,1.7)(1.6,2)
    };

\end{axis}
\end{tikzpicture}
\end{minipage}
\begin{minipage}{4.6cm}
\begin{tikzpicture}
\begin{axis}[
    xlabel={ },
    ylabel={ },
    xmin=0, xmax=4,
    ymin=0, ymax=4,
    xtick={0,1,2,3,4},
    ytick={0,1,2,3,4},
    ymajorgrids=true,
    xmajorgrids=true,
]

\addplot+[
    name path=AA,
    color=black,
    no marks,
    very thick,
    ]
    coordinates {
    (0.1,0.1)(2,3.4)
    };
    
\addplot+[
    name path=BA,
    color=black,
    no marks,
    very thick,
    ]
    coordinates {
    (0.1,0.1)(2,0.1)
    };
    
\addplot[
    color=blue!40,
]
fill between[of=AA and BA];

\addplot+[
    name path=AB,
    color=black,
    no marks,
    very thick,
    ]
    coordinates{
    (2,3.4)(3.9,0.1)
    };
    
\addplot+[
    name path=BB,
    color=black,
    no marks,
    very thick,
    ]
    coordinates {
    (2,0.1)(3.9,0.1)
    };
    
\addplot[
    color=blue!40,
    domain=1.5:4,
]
fill between[of=AB and BB];

\addplot[
    color=black,
    no marks,
    very thick,
    ]
    coordinates {
    (1,1.7)(3,1.7)
    };
    
    \addplot[
    color=black,
    no marks,
    very thick,
    ]
    coordinates {
    (0.6,1)(3.3,1)
    };
    
    \addplot[
    color=black,
    no marks,
    very thick,
    ]
    coordinates {
    (2,1)(2,1.7)
    };
    
    \addplot[
    color=black,
    no marks,
    very thick,
    ]
    coordinates {
    (2,0.1)(1.5,1)
    };

\end{axis}
\end{tikzpicture}
\end{minipage}

\caption{Cones in 3D are determined by a convex subset of the triangle with vertices at \((1,0,0)\), \((0,1,0)\) and \((0,0,1)\). Namely, they are the closure of such a convex set under scalar multiplication. In these plots we draw this triangle in the plane, for visual purposes. \newline
\textit{Left}: The convex set determining \(\dir(\vect{P}_1)\). Directions always contain the interior, hence only the boundary of the triangle is interesting. Only the point \((0,0,1)\), corresponding to pumping only \(z\), is missing. However, this fact suffices such that this convex set is no longer closed, and hence the corresponding cone is not finitely generated, where finitely generated here equivalently means that the convex subset determining this cone is a (closed) polygon. \label{FigureExampleAppendix} \newline
\textit{Middle + Right}: The cone of every full periodic set is finitely generated, i.e.\  its corresponding convex subset is a polygon. Hence every partition of \(\N^3\) into full linear sets \(\vect{S}_1, \dots, \vect{S}_r\) induces a covering of the triangle by polygons/''tiles'' \(T(\vect{S}_i)\), two possible examples depicted here. It is obviously enough to consider the 2D polygons, i.e.\ 3D \(\vect{S}_i\), to cover the triangle. Though this is hard to formalize, and hence we prefer to avoid such an argument in the formal proof. In every part of this tiling into 2D polygons, \(\vect{X}\) is by assumption almost linear, i.e.\ has some set of directions. We will prove these directions to equal the cone corresponding to \(T(\vect{S}_i)\) (Step 2). Hence the fact that the \(T(\vect{S}_i)\) are a complete tiling of the triangle then gives us the desired complete extraction. (Step 3)}
\end{figure}

\subsection{Importance of Complete Extraction vs. having every direction.}

This whole example and especially the proof are built around a fundamental understanding of why we need a complete extraction in Theorem \ref{TheoremEquivalentReducibilityCondition}, and having every direction is not enough. Let us elaborate a bit on this condition. By definition, \(\vect{v}_1\) and \(\vect{v}_2\) being directions in \(\dir(\vect{P}_i)\) means there exist points \(\vect{x}_1\) and \(\vect{x}_2\) such that \(\vect{x}_1 + \N \vect{v}_1 \subseteq \vect{P}_i\), and \(\vect{x}_2 + \N \vect{v}_2 \subseteq \vect{P}_i\). But in case they belong to the same \(\vect{P}_i\), then also \(\vect{x}+\N \vect{v}_1 + \N \vect{v}_2 \subseteq \vect{P}_i\) for some \(\vect{x}\): I.e.\ containing the two lines implies containing a plane between them. The same of course for any finite set \(\vect{F}\) of vectors \(\vect{v}_i\), and then containing \(\vect{x}+\vect{F}^{\ast}\) generated by \(\vect{F}\).

This is fundamentally wrong in case that \(\vect{v}_1\) and \(\vect{v}_2\) are directions in different \(\vect{P}_i\) and \(\vect{P}_j\), even if their union has a hybridization. Let us apply this idea to our example of this section itself, before afterwards generalizing to sections in a partition. Assume for contradiction that \(\vect{X}\) is almost linear. Then, since it has all three directions \((1,0,0)\), \((0,1,0)\) and \((0,0,1)\), it has to contain \(\vect{x}+\N^3\) for some point \(\vect{x}\) by the above argument, i.e.\ has to be reducible. Contradiction. In the formal proof, intuitively, we apply this same argument in every part \(\vect{S}_i\) of the partition, and then use Theorem \ref{TheoremEquivalentReducibilityCondition} to prove that if we contain a shifted version of every \(\vect{S}_i\), then we contain a shifted version of \(\N^3\).

\subsection{Formalizing the Steps}

Proof of step 1: Assume there exists a partition \(\N^3=\vect{S}_1 \cup \dots \cup \vect{S}_r\) into full linear \(\vect{S}_i=\vect{c}_i+\vect{Q}_i\) such that \(\vect{X} \cap \vect{S}_i\) is almost linear if it is non-empty. Hence for every \(i\), there exists a smooth \(\vect{P}_i\) such that \(\vect{b}_i + \vect{P}_i = \vect{X} \cap \vect{S}_i\). By Lemma \ref{LemmaFullThenLessPeriods}, we have \(\vect{P}_i \subseteq \vect{Q}_i\), and then in particular also \(\Fill(\vect{P}_i) \subseteq \Fill(\vect{Q}_i)=\vect{Q}_i\). Hence the full linear set \(\vect{S}_i':=\vect{b}_i + \Fill(\vect{P}_i) \subseteq \vect{S}_i\). Furthermore, \(\vect{X} \cap \vect{S}_i \subseteq \vect{S}_i'\), since \(\Fill(\vect{P}_i)\) overapproximates \(\vect{P}_i\). Hence \(\vect{X} \cap (\vect{S}_i \setminus \vect{S}_i')=\emptyset\). Replacing for every \(i\) the part \(\vect{S}_i\) of the partition by \(\vect{S}_i'\) and (a full linear partition of) \(\vect{S}_i \setminus \vect{S}_i'\), we obtain a partition where \(\vect{X} \cap \vect{S}_i'\) is almost linear with Fill \(\vect{S}_i'\) if it is non-empty.

Proof of step 2: Case 1: Let \(\vect{S}\) full linear with \(\dim(\vect{S}) \leq 2\), we have to prove that \(\vect{X} \cap \vect{S}\) is semilinear. If \(\dim(\vect{S}) \leq 1\), then this is automatic because all almost semilinear sets in dimension 1 are semilinear. 

Otherwise write \(\vect{S}=\vect{c}+\vect{Q}\). We make the simplifying assumption that \(\vect{Q}\) has \(2\) generators, written \((x,y,z)\) and \((x',y',z')\), and do a case distinction on them:

Case 1.1: Neither vector is collinear to \((0,0,1)\): Then for every large \((x_1,y_1,z_1) \in \vect{P}_1\), also \((x_1,y_1,z_1)+\N (x,y,z)+\N (x',y',z') \subseteq \vect{P}_1\). For example, if \((x',y',z')=(0,1,5)\), then every vector with \(x_1 \geq 5\) is big enough for pumping \((x',y',z')\). This is easily obtained from the defining inequality \(z \leq x \cdot y\), since for every increase of \(1\) in \(z\), we pump \(y\) by \(\frac{1}{5}\). Together with \(x_1 \geq 5\), this means the right hand side increases more than the left hand side.

Since both vectors being collinear would imply \(\dim(\vect{S}) \leq 1\), the only subcase left is 

Case 1.2: One vector is collinear to \((0,0,1)\), w.l.o.g. \((x,y,z)=(0,0,1)\): 

Case 1.2.1: If \((x',y',z')\) is an interior vector, i.e.\ \(x' \neq 0, y' \neq 0, z' \neq 0\), then already \(\vect{P}_2 \cap \vect{S}\) is reducible: Namely \(\vect{P}_2\) basically states that \(x,y\) are not more than exponentially different (which only becomes easier by pumping both of them) and \(z \geq \log(x)\), which follows from every pumping vector increasing \(z\).

Case 1.2.2: Otherwise \((x',y',z')\) has some zero component. If it is either \(x'\) or \(y'\), then since only one of them is getting pumped, \(\vect{P}_2 \cap \vect{S}\) projected to \(\N^2\), the first two components, is finite. Hence \(\vect{P}_2 \cap \vect{S}\) is at most 1-dimensional, and hence semilinear. On the other hand, \(\vect{P}_1 \cap \vect{S}\) is semilinear, since one of the two coordinates \(x,y\) is now a constant, and \(z \leq xy\) hence a semilinear condition.

Case 1.2.3: Hence \(z'=0\), and both \(x' \neq 0\) and \(y' \neq 0\). Since we either pump both or none of the first two coordinates, the conditions \(y \leq 2^x\) and \(x \leq 2^y\) of \(\vect{P}_2\) are automatically fulfilled for almost all points in \(\vect{S}\). Furthermore, almost all points will fulfill either the condition \(z \leq x \cdot y\) or \(x \leq 2^z \Leftrightarrow z \geq \log(x)\). Namely, this shape is similar to the ``above logarithm'' union ``below parabola'' set in the middle of Figure \ref{FigureIntuitionAlmostHybridlinear}. Hence \(\vect{X} \cap \vect{S}\) is reducible.

Case 2: \(\dim(\vect{S})=3\). Here we have to show that \(\dir(\vect{X} \cap \vect{S})=\Q_{\geq 0} \vect{Q}\). The main observation for this case is that for every bound \(B \in \N\) there exists a point \(\vect{v} \in \vect{P}_1 \cap \vect{P}_2 \cap \vect{S}\), such that every coordinate of \(\vect{v}\) is \(\geq B\). With this fact, the argument is mainly a repeat of the argument for \(\vect{X}\) itself, and we skip this case.

Proof of step 3: Write \(\vect{X} \cap \vect{S}_i=\vect{b}_i + \vect{P}_i\) for all \(i\) where it is non-empty, with w.l.o.g. \(\vect{S}_i=\vect{b}_i+\Fill(\vect{P}_i)\) by step 1. Since \(\vect{b}_i + \vect{P}_i\) is reducible by step 2, we have \(\dir(\vect{P}_i)=\Q_{\geq 0} \Fill(\vect{P}_i)\) by Proposition \ref{PropositionAlmostPeriodicityPeriodicSet}. In particular, the sets of directions are finitely generated cones. We have \(\bigcup_{i=1}^r \Q_{\geq 0} \Fill(\vect{P}_i)=\Q_{\geq 0}^3\), since \(\N^3=\vect{S}_1 \cup \dots \cup \vect{S}_r\). [This corresponds to the intuition that the \(\vect{S}_i\) give rise to a complete tiling.] We trivially have \(\vect{X}=\vect{X} \cap \N^3=\bigcup_{i=1}^r \vect{b}_i+\vect{P}_i\). 

We would want to simply use Theorem \ref{TheoremEquivalentReducibilityCondition}, since the \(\dir(\vect{P}_i)\) are all finitely generated and hence themselves form a complete extraction. However, this might not be an almost hybridlinear representation. Hence similar to the proof of Theorem \ref{TheoremEquivalentAlmostHybridlinearCondition}, let \(\vect{P}\) be smooth such that \(\vect{X}=\bigcup_{i=1}^r \vect{b}_i + (\vect{P}_i + \vect{P})\) is an almost hybridlinear representation. We have \(\dir(\vect{P}_i + \vect{P}) \subseteq \Q_{\geq 0}^3\) for all \(i\), since all periodic sets only contain non-negative vectors. Together with the above, we obtain \(\bigcup_{i=1}^r \dir(\vect{P}_i + \vect{P})=\Q_{\geq 0}^3=\bigcup_{i=1}^r \dir(\vect{P}_i)\). Remember that these later cones are finitely generated. Hence the new almost hybridlinear representation has a complete extraction, and \(\vect{X}\) is therefore reducible by Theorem \ref{TheoremEquivalentReducibilityCondition}. This finishes the proof by contradiction.

\end{document}